\documentclass[letterpaper, journal,10 pt]{IEEEtran} 

\usepackage{graphicx}
\usepackage{textcomp}
\usepackage{epsfig} 
\usepackage{epsfig,color,amsmath,cite}

\usepackage{amsthm} 
\usepackage{amsmath}    

\usepackage{enumitem} 
\setlist[itemize]{leftmargin=*}

\usepackage{bm}
\usepackage{epstopdf}
\usepackage{amssymb}
\usepackage{url}

\usepackage{multirow}
\usepackage{hhline}
\usepackage{booktabs}
\usepackage[linesnumbered,boxed,commentsnumbered,ruled,vlined,longend]{algorithm2e}
\usepackage{comment}
\newcommand{\eps}{\varepsilon}

\DeclareMathOperator*{\maximize}{maximize}

\DeclareMathOperator*{\subjectto}{subject\ to}

\makeatother
\DeclareMathAlphabet\mathbfcal{OMS}{cmsy}{b}{n}

\newtheorem{theorem}{Theorem}
\newtheorem{mydef}{Definition}

\newtheorem{mypro}{Property}
\newtheorem{myrem}{Remark}
\newtheorem{asmp}{Assumption}

\newtheorem{myprs}{Proposition}



\usepackage{stackengine}

\newcommand{\mat}[1]{\boldsymbol{#1}}

\newcommand{\bmat}[1]{\begin{bmatrix} #1 \end{bmatrix}}

\providecommand{\eye}{\mat{I}}
\providecommand{\mA}{\ensuremath{\mat{A}}}

\providecommand{\mG}{\ensuremath{\mat{G}}}

\newcommand{\m}{\boldsymbol}
\allowdisplaybreaks[4]
\usepackage[colorlinks = true,
linkcolor = blue,
urlcolor  = black,
citecolor = blue,
anchorcolor = blue]{hyperref}

\newcommand{\mc}[1]{\mathcal{#1}}
\newcommand{\mbb}[1]{\mathbb{#1}}
\newcommand{\mr}[1]{\mathrm{#1}}
\usepackage[framemethod=TikZ]{mdframed}
\mdfdefinestyle{MyFrame}{%
	linecolor=black,
	outerlinewidth=1.25pt,
	roundcorner=1.25pt,
	innerrightmargin=5pt,
	innerleftmargin=5pt,}
	
\usepackage[noabbrev]{cleveref}

\usepackage{mathtools}

\DeclarePairedDelimiter\abs{\lvert}{\rvert}%
\DeclarePairedDelimiter\norm{\lVert}{\rVert}%

\makeatletter
\let\oldabs\abs
\def\abs{\@ifstar{\oldabs}{\oldabs*}}
\let\oldnorm\norm
\def\norm{\@ifstar{\oldnorm}{\oldnorm*}}
\makeatother

\usepackage[english]{babel}
\usepackage[utf8]{inputenc}
\usepackage[super]{nth}

\usepackage{graphicx}
\usepackage{float}
\usepackage[caption = false]{subfig}

\usepackage{array}
\usepackage{threeparttable}

\usepackage[english]{babel}
\usepackage[utf8]{inputenc}
\usepackage[super]{nth}

\usepackage{pifont}

\usepackage{mathtools}

\usepackage{mleftright}
\usepackage{xparse}

\NewDocumentCommand{\evaluat}{sO{\big}mm}{%
	\IfBooleanTF{#1}
	{\mleft. #3 \mright|_{#4}}
	{#3#2|_{#4}}%
}

\usepackage{lettrine} 

\usepackage{balance} 

\usepackage{tikz}
\usetikzlibrary{shapes.geometric, arrows.meta, decorations.pathmorphing, decorations.markings, calc, 3d}
\usetikzlibrary{positioning, fit, backgrounds}
\tikzset{
	tangent/.style={decoration={
			markings,mark=at position #1 with {
				\coordinate (ta) at (0,0);
				\coordinate (tb) at (0.1,0);
			}
		},postaction=decorate},
	tangent/.default=0.5
}

\newcommand{\parstart}[1]{\noindent \textbf{#1.}\;}

\newcommand{\nocol}[4]{\left\{ {#1}_{#2}\right\}_{{#2}={#3}}^{#4}}

\newcommand{\Rn}[1]{\mathbb{R}^{#1}}

\newcommand{\titlebf}[1]{\title{\Large \vspace{0.7cm} \LARGE \centering {\textsc{\textbf{#1}}}}}

\usepackage{caption} 
\titlebf{Stability and Uncertainty Propagation in Power Networks: A Lyapunov-based Approach with applications to Renewable Resources Allocation}

\author{Mohamad H. Kazma, \textit{Graduate Student Member, IEEE} and Ahmad F. Tah$\text{a}^{\diamond}$, \textit{Member, IEEE} \vspace{-0.2cm}
	\thanks{ 
		$^\diamond$Corresponding author. This work is supported by National Science Foundation under Grants 2152450 and 2151571. The authors are with the Civil $\&$ Environmental Engineering and Electrical $\&$ Computer Engineering Departments, Vanderbilt University, 2201 West End Ave, Nashville, Tennessee 37235. Emails: mohamad.h.kazma@vanderbilt.edu, ahmad.taha@vanderbilt.edu.}
}

\begin{document}
\newdimen\origiwspc%
\newdimen\origiwstr%
\origiwspc=\fontdimen2\font
\origiwstr=\fontdimen3\font

\maketitle
\markboth{IEEE Transactions on Power Systems, In Press, February 2025}{}

\begin{abstract}
The rapid increase in the integration of intermittent and stochastic renewable energy resources (RER) introduces challenging issues related to power system stability. Interestingly, identifying grid nodes that can best \textit{support} stochastic loads from RER, has gained recent interest. Methods based on Lyapunov stability are commonly exploited to assess the stability of power networks. These strategies approach quantifying system stability while considering: \textit{(i)} simplified reduced order power system models that do not model power flow constraints, or \textit{(ii)} data-driven methods that are prone to measurement noise and hence can inaccurately depict stochastic loads as system instability. In this paper, while considering a nonlinear differential algebraic equation (NL-DAE) model, we introduce a new method for assessing the impact of uncertain renewable power injections on the stability of power system nodes/buses. The identification of stable nodes \textit{informs} the operator/utility on how renewables injections affect the stability of the grid. The proposed method is based on optimizing metrics equivalent to the Lyapunov spectrum of exponents; its underlying properties result in a computationally efficient and scalable stable node identification algorithm for renewable energy resources allocation. The developed framework is studied on various standard power networks.
\end{abstract}
\vspace{-0.2cm}
\markboth{IEEE Transactions on Power Systems, In Press, February 2025}{a}
\section{Introduction and paper contributions}\label{sec:Introduction}
 \setlength{\textfloatsep}{0pt}
\noindent \lettrine[lines=2]{P}{ower} systems are becoming more reliant on renewable energy resources (RER) for power supply, in a much needed effort to decarbonize our electricity systems. From a stability perspective, the uncertainty and intermittent generation of RER may impede power system operations. {This results from the intermittent perturbations acting on the synchronous states~\cite{Li2017a}.}
As such, studying the transient stability and resilience of power systems with increased amounts of integrated RER has become essential. {A fundamental stability analysis approach} is typically provided by simulating the impact of RER integration through load disturbances and inertia changes, and then assessing the system settling time.

{Transient stability analysis entails the identification of potential disturbances that might lead to system instabilities~\cite{Bosetti2018}.} To assess system stability, Lyapunov stability methods are widely adopted~\cite{Chen2023}. The methods are divided into two categories: \textit{(i)} a direct stability method that is based on quantifying a Lyapunov energy function~\cite{Isbeih2019} and \textit{(ii)} a \textit{Lyapunov exponents} (LEs) method that is based on characterizing infinitesimal separation rates of system trajectories~\cite{Bosetti2018}. The former method is based on quantifying an energy function that is indicative of system stability. The latter method is adapted from the field of chaos and was first introduced by~\cite{Liu1994} in the context of power networks for studying transient chaotic swings. {The aforementioned direct stability method is well-developed; however, it becomes complex to quantify when considering nonlinear dynamical systems.} On the other hand, quantifying LEs of a nonlinear system arises more naturally~\cite{Czornik2019}. 

For power system stability, the maximal Lyapunov exponent (MLE)  is able to depict general characteristics of transient stability~\cite{Zong2023}. {A negative MLE sign indicates stability}, whereas a positive MLE informs that the system is unstable. Theoretically, LEs measure the average exponential rate of convergence and divergence of a pair of nearby trajectories in a multidimensional state-space. {The exponents are key to determining stability since the exponents are invariant to any initial condition perturbation within the same stability region~\cite{Shen2023}.} That being said, the main objective of this paper is to provide a framework for quantifying stability based on metrics related to the LEs of the system. The exponents are computed after an uncertain RER power injection is applied to a bus. The method quantifies power system stability and uncertainty propagation from negative loads that model injections of uncertain renewables. As such, the proposed method \textit{informs} the operator about the impact of allocating renewable injections, at network buses, on the overall power grid stability. 

\parstart{Literature Review} 
{The LE-based approach} for power system stability and its applications in quantifying network characteristics has gained recent interest. The literature can be divided into (i) simply quantifying the stability of a power network, and (ii) identifying network properties. {Power system stability is usually studied in terms of voltages, angles, and frequencies~\cite{Sobbouhi2021}.} A method based on LEs~\cite{Yan2011}, is utilized to study the \textit{swings} that lead to instability by monitoring the rotor angle of generator machines. In~\cite{Wadduwage2013}, the transient stability of power systems post-fault is studied by analyzing the system's LEs. Voltage stability is considered from phasor measurement units (PMU) data in~\cite{Dasgupta2013}. The study considers a model-free approach that monitors voltage data and then estimates the systems LEs for stability predictions. In~\cite{Bosetti2018}, global system stability is studied by computing LEs and their corresponding Lyapunov vectors. The system machines are modeled according to the swing equations {and then the proposed approach is demonstrated on a $9$-bus network.} A model-free estimation framework is considered for assessing rotor angle stability in~\cite{Wei2018}. The proposed method relies on least-squares state {estimation to compute rotor angle estimates from PMU measurements.} {Similarly,~\cite{Zong2023} considers rotor angle stability from MLE estimation while utilizing a recursive least squares estimation framework that is less time-consuming.}

The LE-based stability methods outlined above consider the power system dynamics by \textit{(i)} modeling simplified and reduced-order power machine models or \textit{(ii)} utilizing model-free approaches that rely on PMU measurement data. Moreover, these stability methods consider only voltage stability and rely on dynamic state estimation methods for estimating rotor angle stability. The main drawback of such methods are as follows. Simplified machine models that do not model the both differential equations and algebraic constraints are not equipped to model RER loads along with their uncertainty. The reason is that renewables injections do not explicitly appear in the differential equations but appear in the algebraic power flow constraints. On the other hand, a direct model-free approach is only appropriate for disturbance and noise free systems~\cite{Escot2023}. The reason is that any disturbance or uncertainty would be wrongly considered as chaotic behavior when computing LEs. As such, there is no clear way to explicitly delineate how renewable injections would impact grid dynamics and the different types of stability. Moreover, when relying on estimation frameworks for rotor angle stability, estimation errors can also be incorrectly interpreted as chaos in the system. 
Accordingly, in this paper we compute a power system's LEs by considering a model-based approach that models the complete nonlinear differential algebraic power system model and therefore depicts both differential and algebraic states.

The second branch of research that goes beyond simply studying the stability of power systems and has recently gained interest,{ is the application of LE methods that enable network properties identification.} The studies~\cite{Khaitan2013, Rashidi2016} both utilize a LE-based algorithm to identify coherent generator groups that generate synchronous swings. The effect of chaos from wind power is predicted from a model decomposition approach that is based on LEs computations in~\cite{Safari2017}. {An optimal PMU allocation approach for full network observability and improved transient stability is proposed in~\cite{Rashidi2016a}.} The approach relies on LEs to identify the critical buses that provide the most useful information for each system node. {Nonetheless, influential nodes in power systems play a critical role in influencing the structure and dynamics of the network~\cite{Rajeh2023}.} The allocation of RER is proposed in~\cite{Yoon2020}. {The method offers a heuristic placement algorithm that estimates LEs based on voltage measurements and then ranks the candidate buses according to the LEs computed from voltage stability.} However, the framework proposed in~\cite{Yoon2020} offers a candidate node rank that \textit{(i)} is 
based on model-free LE estimation, where instability can result from measurement noise, \textit{(ii)} computes LEs while considering chaos resulting {from only voltage stability,} and \textit{(iii)} quantifies node stability by measuring the LEs at that node while not taking into account uncertainty propagation onto the other nodes.

\noindent \textbf{Paper Contributions.}\; 
Motivated by the aforementioned limitations in the recent literature, the main contributions of this work are summarized as follows.
\begin{itemize}[leftmargin=*]
	\item We introduce a framework that enables the node stability assessment of an NL-DAE representation of a power network. In specific, we provide quantitative stability measures for identifying stable and unstable nodes in a power network from a Lyapunov stability perspective. The proposed stability measures are based on the Lyapunov characteristic spectrum of exponents, which are computed for a system under RER load perturbations.
	These stability metrics allow for (i) computing the overall stability by considering frequency, voltage, and rotor angle stability and (ii) {quantifying renewables uncertainty propagation, applied to a single node, on the overall power network stability.}
	\item Based on the introduced stability framework, we propose an algorithm for identifying optimal stable nodes that can accommodate renewables injections. This approach allows to inform the operator where to practically allocate renewables while maintaining the overall stability of the power network that is inundated with uncertain fluctuations from RER injections. The result is an ordered set of nodes that {reflects stable nodes with minimal impact on overall power system stability.}
\end{itemize}	

In addition to the theoretical contributions of this paper, we demonstrate the proposed framework on two distinct power grids: a $4^{\text{th}}$-order system and a $9^{\text{th}}$-order system that includes two PV solar plants. The stability framework illustrates the optimality of the identified stable nodes under uncertain RER perturbations, modeled either by inducing perturbed negative RER injections or by perturbing the solar irradiance on the PV plants.

We note that in this paper we do not review methods for RER allocation or sizing. Such methods can consider different operational objectives (efficiency maximization, and  power loss and cost minimization). Instead, we are concerned with quantifying the effect of RER allocation on the overall stability of a power network and, thus, provide a method to inform operators regarding the practical allocation of RER from a stability perspective. Readers are referred to~\cite{Bazrafshan2019} for a more comprehensive review on RER allocation and sizing.

\parstart{Broader Impacts}The model-based stable node identification framework enables the utilization of general NL-DAE dynamical models of power grids. As such, based on the joint modeling of differential and algebraic system states, we can capture and model uncertain RER power generation from sources such as PV plants and wind farms. This allows for quantifying the uncertainty propagation and assess its overall stability implications on the entire power grid. Furthermore, system operators can consider different stability aspects—such as voltage, frequency, and rotor angle stability—by considering the proposed parameterized stability measures. This, in turn, results in a framework that provides system operators with quantifiable measures regarding which network buses maintain grid stability in the presence of uncertain RER inputs{, whether such RERs are already integrated into the network or to yet be allocated.}

%

\parstart{Notation} Let $\mathbb{N}$, $\mathbb{R}$, $\Rn{n}$, and $\Rn{p\times q}$ denote the set of natural numbers, real numbers, real-valued row vectors with size of $n$, and $p$-by-$q$ real matrices respectively. The symbol $\otimes$ denotes the Kronecker product. The cardinality of a set $\mc{N}$ is denoted by $|\mc{N}|$. The operators $\mr{det}(\m{A})$ returns the determinant of matrix $\m{A}$,  $\mathrm{trace}(\m{A})$ returns the trace of matrix $\m{A}$. The operator $\{\m{x}_{i}\}_{i=0}^{\mr{N}}  \in \Rn{\mr{N}n}$ constructs a column vector that concatenates vectors $\m{x}_i \in \Rn{n}$ for all $i \in \{0, 1,\ldots, \mr{N}\}$.

\parstart{Paper Organization} The paper is organized as follows:~Section~\ref{sec:prelims} provides the power system model.~Section~\ref{sec:IdentStability} provides preliminaries on LEs computation.~Section~\ref{sec:main} introduces the proposed stable node identification framework and the optimal RER allocation problem for NL-DAEs. The numerical results are presented in Section~\ref{sec:casestudy}, and Section~\ref{sec:summary} concludes this paper. 
\section{Problem formulation}\label{sec:prelims}
In this section, we present the underlying NL-DAE power system model and represent the full model as a nonlinear ordinary differential equation (NL-ODE) model. The resulting NL-ODE system, that includes the algebraic constraints, is discretized using the {trapezoidal implicit (TI)} discretization method.
\subsection{NL-DAE power system state-space model}\label{subsec:NDAE}
We consider a NL-DAE formulation of a power system in semi-implicit form~\cite{Milano2016}. To this end, we borrow the notation from~\cite{Nugroho2022}. A power system $(\mc{N} ,\mc{E})$ can be represented graphically, where $\mc{E} \subseteq \mc{N} \times \mc{N}$ is the set of transmission lines, $\mc{N} = \mc{G} \cup \mc{L} \cup \mc{R} $ is the set of all buses in the network. Set $\mc{G}=\{1, \ldots, G\}$ denotes the set of synchronous generators, $\mc{R}=\{1, \ldots, R\}$ represents the set of buses with solar power plants or renewable load injections, and $\mc{L}$ collects set of load buses. The number of buses within the network is $N := |\mc{N}|$, while the number of generator, load and renewable buses are $G := |\mc{G}|$, $L := |\mc{L}|$ and $R := |\mc{R}|$. The physics-based components of the electromechanical transients can be written in a semi-implicit NL-DAE form as follows
\begin{subequations}\label{eq:semi_NDAE_rep}
	\begin{align}
		\textit{generator dynamics}:	\;\; \dot{\m x}_{d} &=  \m{f}(\m x_d,\m x_a, \m u) \label{X_d},\\ 
		\textit{algebraic constraints}:	\;\; \m 0 & = \m{g}(\m x_d, \m x_a),\label{X_a}
	\end{align} 
\end{subequations}
where the dynamic states of the synchronous machines, PV plants, motor loads and loads can be represented as $\m{x}_{d} := \m{x}_{d}(t) \in \Rn{n_d}$, the algebraic states can be defined as $\m{x}_{a} := \m{x}_{a}(t) \in \Rn{n_a}$, and the control input of the system can be defined as $\m{u}:=\m{u}(t)\in \Rn{n_u}$. Here, $n_d$ and $n_a$ define the size of the differential and algebraic states, respectively, and $n_u$ defines the size of the input state vector $\m{u}$. Nonlinear mapping function $\m{f}(\cdot)$ can represent the synchronous machine, excitation system, governor, turbine dynamics, and solar power plant dynamics, where $\m{f}(\cdot): \Rn{n_d}\times\Rn{n_u}\times\Rn{n_a}\rightarrow\Rn{n_d}$. Nonlinear mapping function $\m{g}(\cdot)$ depicts the algebraic constraints and the power balance equations, where $\m{g}(\cdot): \Rn{n_d}\times\Rn{n_a}\rightarrow\Rn{n_a}$.

In this manuscript, we consider two power system models: \textit{(i)} the standard two-axis $4^{\text{th}}$-order transient model of a synchronous generator~\cite[Ch. 7]{Sauer2017}, and \textit{(ii)} the comprehensive $9^{\text{th}}$-order transient model of a synchronous generator interconnected with a $12^{\text{th}}$-order grid-forming PV plant model for $i \in \mc{R}$~\cite{Roy2023}. The model describes the DC link and PV array dynamics, DC/AC converter and LCL filter dynamics, and voltage/current regulators models. The physics-based components of the two power network models are summarized in Appendix~\ref{apndx:power_model} and Appendix~\ref{apndx:power_model2}.

\subsection{NL-DAEs: existence and uniqueness of a solution}\label{subsec:solutionexist}
The uniqueness and existence of a solution for NL-DAEs, as compared to NL-ODEs, for any initial condition $(\m{x}_0,\m{u}_0)$ is not always guaranteed~\cite{Guo2021}. The existence and uniqueness of solutions to NL-DAEs holds true if and only if the system~\eqref{eq:semi_NDAE_rep} is \textit{strangeness-free}~\cite[Hypothesis 4.2]{Kunkel2006}, whereby the strangeness index is equal to zero. For brevity, refer to
~\cite{Kunkel2006} for the detailed hypothesis that defines the strangeness index of NL-DAEs. We note here that this index is a generalization of the differentiation index of DAEs, which is defined as follows.
\begin{mydef}\label{def:inde(1)}
	The differentiation index~\cite{Kunkel2006} 
	of nonlinear or linear DAE systems is defined as the number of times the algebraic equations are differentiated to obtain a set of ODEs.
\end{mydef}
A linear DAE can be written as $\m{E}\dot{\m{x}} = \m{A}\m{x} + \m{B}\m{u}$. This representation can be formulated by linearizing the system~\eqref{eq:semi_NDAE_rep} around an operating point 
where the state vector $\m{x}$ is defined as $\m{x}:= \left[\m{x}_d,\; \m{x}_a\right]^{\top}\in \Rn{n_d+n_a}$. The singular mass matrix $\m{E} \in \mathbb{R}^{n_d+n_a}$ has ones on its diagonal for the differential equations and zeros for the algebraic equations. The time-invariant state-space matrices are defined as $\m{A} :=\left[\m{A}_d, \; \m{A}_a\right]^{\top}  \in \mathbb{R}^{n_d+n_a}$ and $\m{B}:= \left[\m{B}_d\right]\in \mathbb{R}^{n_u}$. Here $\m{A}_d \in \Rn{n_d}$ represents the linearized differential equations state-space matrix and $\m{A}_a \in \Rn{n_a}$ the linearized algebraic equations state-space matrix.

We note that a linear DAE system is of index one if and only if the system is regular. Regularity is considered an essential property for DAEs. This property ensures the existence of consistent unique solutions for every initial condition $(\m{x}_0,\m{u}_0)$~\cite{Grob2016}.
\begin{mydef}\label{def:regular}
	A linearized DAE around an initial state $\m{x}_0$ is regular if and only if there exists $s \in \mathbb{C}$ such that
	$
	\mr{det}(s\m{E} - \m{A}) \neq 0.
	$ Regularity is therefore characterized by the matrix pair $(\m{E}, \m{A})$.
\end{mydef}
A DAE power network model is regular if and only if there exists a path whereby each load follows this path to a generator bus~\cite{Grob2016}. 
This holds true for the IEEE case networks considered in the case studies section in this work. The regularity and differentiation index of linearized power system~\eqref{eq:semi_NDAE_rep} are assessed in~\cite{Nugroho2022}. It is shown that the power system has a differentiation index of one and is regular;~see~\cite{Nugroho2022} for additional information. Such conditions guarantee that for each consistent initial condition, a unique solution always exists. With the evidence that the linearized DAE system, we can now infer that the NL-DAE is strangeness-free without the rigorous proof. Hence, for certain rank conditions every regular and linear DAE with sufficiently smooth $(\m{E}, \m{A})$ matrix pair satisfies the hypothesis regarding the vanishing strangeness index, i.e., the system is strangeness-free;~see~\cite{Kunkel2006}. {In order to ensure} that there exists a solution to the NL-DAE power system~\eqref{eq:semi_NDAE_rep}, the following assumption holds true throughout this paper. 
\begin{asmp}\label{assump:index}
	The NL-DAE~\eqref{eq:semi_NDAE_rep} is considered strangeness-free, regular and of differentiation index one. As such, a unique solution exists for any initial conditions $(\m{x}_0,\m{u}_0)$.
\end{asmp}
The aforementioned assumption implies that the partial derivatives of the differential and algebraic state-space functions~$\m{f}(\cdot)$ and $\m{g}(\cdot)$ with respect to $\m{x}_{d}$ and $\m{x}_{a}$ are non-singular~\cite{Kunkel2006}. We note here that this assumption is mild and holds true for the power system cases considered in the numerical studies section. 
In practical terms, the numerical solvability of the DAE system~\eqref{eq:semi_NDAE_rep} directly implies the existence and uniqueness of solution for any input and consistent initial conditions~\cite{Crow2015}. The time-domain simulations presented in Section~\ref{sec:casestudy} provide explicit proof and validate Assumption~\ref{assump:index} for the NL-DAE power system~\eqref{eq:semi_NDAE_rep}.
\subsection{NL-DAE to NL-ODE system transformation}\label{subsec:Transformation}
NL-DAEs are considered stiff dynamical systems with time constants that can span several orders of magnitude. 
The algebraic constraints are considered to exhibit null time constants. Stability and control of such NL-DAE systems are limited, while that of NL-ODE systems is not. Typically, {an ODE system representation is considered for power system dynamics modeling by either} (i) neglecting the algebraic constraints or (ii) relying on a decoupled model~\cite{Nugroho2022}. This limits the ability to assess the transient stability of power system, in particular, power networks that include RERs. With that in mind, we consider a structure preserving  
NL-DAE to NL-ODE transformation of the power system model~\eqref{eq:semi_NDAE_rep}. This transformation allows for the complete modeling of the differential and algebraic equations~\eqref{eq:differential_dynamics}–\eqref{eq:power_balance}, while representing them as a set of ODEs.

The transformation relies on applying {the inverse function theorem} (IFT)~\cite[Theorem 3.3.1]{Krantz2013} to resolve the algebraic constraints into ODEs. This requires the differentiation of the algebraic constraints in~\eqref{X_a} with respect to time variable $t$. With that in mind, the NL-DAE system~\eqref{eq:semi_NDAE_rep} can be rewritten as 
\begin{subequations}\label{eq:semi_NDAE_rep-ODE}
	\begin{align}
		\dot{\m x}_{d} &=  \m{f}(\m x_d,\m x_a, \m u), \label{X_d_ode} \\
		\dot{\m{x}}_a & = \tilde{\m{g}}(\m{x}_d,\m{x}_a,\m{u})= -(\mG_{\m{x}_a})^{-1}\mG_{\m{x}_d}\m{f}(\m{x}_d,\m{x}_a,\m{u}),~\vspace{-0.2cm}\label{X_a_ode}
	\end{align} 
\end{subequations}
where matrices defined as $\mG_{\m{x}_a} :=\frac{\partial \m g(\m{x}_d,\m{x}_a) }{\partial \m{x}_a}$ $\in \Rn{n_{a}\times n_{a}}$  and $\mG_{\m{x}_d} :=\frac{\partial \m g(\m{x}_d,\m{x}_a) }{\partial \m{x}_d}$ $\in \Rn{n_{a} \times n_{d}}$ are the Jacobian matrices of the algebraic constraints with respect to $\m{x}_a$ and $\m{x}_d$, respectively. 

The NL-ODE system~\eqref{eq:semi_NDAE_rep-ODE}, obtained using the IFT method, depicts the full dynamic and algebraic relationships that are inherent to the system while also being modeled as a set of ODEs. We note that the resulting NL-ODE system representation of the algebraic constraints~\eqref{X_a} depends on~\eqref{X_d}; it therefore depends on control input $\m{u}$.
Readers are referred to~\cite{Kazma2023b} for a more though discussion regarding the validity and accuracy of the transformed NL-DAE model.
\subsection{Discrete-time modeling of NDAE power systems}\label{subsec:Discretization}
Nonlinear systems, in particular DAE power system models, exhibit stiff dynamics. {Network systems} with stiff dynamics can be characterized by time constants on local nodes or subsystems that have a significant contrasting magnitude. The discretization method of choice relies on the stiff dynamics and the desired accuracy of the discretization. {In practice, to obtain a stable and computationally efficient solution,} power systems that exhibit transient conditions are solved numerically using implicit discretization methods~\cite{Crow2015}. Explicit methods such as the implicit Runge-Kutta (IRK) method cannot deal properly with stiff dynamics~\cite{Milano2016}. 

{In power system discrete-time modeling,} the implicit discretization methods utilized are as follows: \textit{(i)} backward Euler (BE) method~\cite{Milano2016}, \textit{(ii)} backward differential formulas (BDF) known as Gear's method~\cite{Gear1971}, and \textit{(iii)} trapezoidal implicit (TI) method~\cite{Milano2016}. 
In~\cite{Kazma2023b}, BDF and TI methods are investigated for the dynamical system in~\eqref{eq:semi_NDAE_rep}. The results show an accurate depiction of the transient dynamical states. A such, for the purpose of characterizing the LEs of the NL-DAE system~\eqref{eq:semi_NDAE_rep-ODE}, we refer to the use of the TI discretization method. The main advantage of implicit TI method is that it can handle a large class of stiff nonlinear dynamical systems with large discretization time steps.
To such end, the nonlinear power system dynamics~\eqref{eq:semi_NDAE_rep-ODE} can be rewritten in discrete-time as 
\begin{subequations}\label{eq:disc_ssm_NDAE}
	\begin{align}
		 \m{x}_{d,k}  &= \m {x}_{d,k-1} + \tilde{h} \left(\m{f}(\m{z}_{k})+\m{f}(\m{z}_{k-1})\right),\\
		\m{x}_{a,k} &=  \m {x}_{a,k-1} + \tilde{h} \left(\tilde{\m{g}}(\m{z}_{k})+\tilde{\m{g}}(\m{z}_{k-1})\right),
	\end{align}
\end{subequations}
where vector $\m{z}_{k} :=[\m{x}_{d,k}, \;\m{x}_{a,k},\;\m{u}_{k}]^{\top}\in \Rn{n_d+n_a+n_u}$ and $\m{x}_{k} :=[\m{x}_{d,k},\; \m{x}_{a,k}]^{\top}\in \Rn{n_d+n_a}$ for time step $k$. The discretization time step size $\tilde{h}$ is defined as $\tilde{h} :=0.5h $, where $h$ is the simulation time step size.

The Newton-Raphson (NR) algorithm is implemented to solve a system of implicit discrete-time NL-DAEs;~see~\cite{Milano2022}. 
For time-domain simulations, the algorithm computes the unknown future states $\m{x}_{k}$ using an iterative framework until solution convergence, and then it advances to the next time step. Refer to Appendix~\ref{apndx:NR-method} for a detailed description of the NR method.  
Moving forward we shall imply the dependence on control input, that is, $\m{f}(\m{z}_k):= \m{f}(\m{x}_k)$ and $\tilde{\m{g}}(\m{z}_k):= \tilde{\m{g}}(\m{x}_k)$.
The discrete-time nonlinear power system~\eqref{eq:disc_ssm_NDAE} can be written succinctly as 
\begin{equation}\label{eq:disc_NDAE-ODE} 
	\m{x}_k =  \m{x}_{k-1} + \tilde{h}\left[\hat{\m{f}}(\m{x}_{k}) + \hat{\m{f}}(\m{x}_{k-1})\right],
\end{equation}
where nonlinear function $ \hat{\m{f}}(\m{x}_{k}):= \left[
\m{f}(\m{x}_{k}), \; \tilde{\m{g}}(\m{x}_{k})\right]^{\top}\in  \Rn{n}$, such that $n:= n_d +n_a$ represents the differential and algebraic states.  
\section{A Lyapunov-based stability approach}\label{sec:IdentStability}
This section provides a brief background on the stability of dynamical systems from a chaos and ergodic theory perspective; it describes the rationale and practical implementation of LEs for stability of nonlinear systems.
\subsection{Lyapunov stability: maximal Lyapunov exponent}\label{subsec:VarDyn}
A well-known method for assessing the stability of a nonlinear system’s trajectory stems from Lyapunov’s indirect method of stability~\cite{Leszczynski2022}. This method, common to the fields of chaos and ergodicity, involves computing the Lyapunov spectrum of exponents. These exponents represent the exponential convergence or divergence of nearby perturbed trajectories of an attractor in the state-space~\cite{Wei2018}. 
{LEs, denoted by $\m{\lambda}$,} are an indication of the asymptotic behavior of dynamical systems~\cite{Sun2012}. As such, this method presents a tool to assess the systems sensitive dependence of initial conditions, i.e., analyze the system's response to {transients induced by perturbed initial conditions.} To calculate the LEs of nonlinear power system models, the use of the dynamical system’s variational form becomes necessary~\cite{Hayes2018a}.

Consider two nearby trajectories $\m{x}_{k}$ and $\m{x}_{k} + \m\delta{\m{x}_k}$ resulting from initial states $\m{x}_{0}$ and $\m{x}_{0} + \m\delta{\m{x}_0}$. Note that $\m \delta{\m{x}_0}\in \Rn{n}$ is an infinitesimal perturbation $\eps>0$ to initial conditions $\m{x}_0$ and its exponential decay or growth for $k \in \{0\;, 1\; ,\;\ldots\;,\mr{N}-1\}$ is denoted as $\m \delta\m{x}_k\in \Rn{n}$. The initial conditions are the operating conditions obtained by solving the power flow equations. Consequently, the uncertainty from renewables are modeled by perturbing the initial loads in power flow equations. {Note that $\mr{N}$ is the discrete-time simulation interval, defined as $\mr{N}:=t/h$.} To that end, the variational form of the discrete-time nonlinear system~\eqref{eq:disc_NDAE-ODE} can be written as
  \begin{equation}\label{eq:DiscVarState}
 	\m{\delta}\m{x}_{k} =
 	\m{\Phi}_{0}^{k}(\m{x}_0) \m{\delta}\m{x}_0,
 \end{equation}
where $\m{\Phi}_{0}^{k}(\m{x}_0):=\left(\m{I}_{n}  +\tfrac{\partial\tilde{\m{f}}( \m{x}_{k},\m{x}_{k-1})}{\partial\m{x}_{k-1}}\right)\tfrac{\partial\m{x}_{k-1}}{\partial\m{x}_{0}} \in \Rn{n \times n}$ defines the variational mapping function, such that $\m{\Phi}_{0}^{0}(\m{x}_0)=\m{I}_{n}$ and matrix $\m{I}_{n} \in \Rn{n\times n}$ is an identity matrix. Readers are referred to~\cite{Kawano2021} 
for the derivation of the variational system~\eqref{eq:DiscVarState}. The variational system~\eqref{eq:DiscVarState} depicts how small system disturbances evolve along the system trajectory. {That is,~\eqref{eq:DiscVarState} allows the quantification of the propagation of uncertain load disturbances within the power system.} For ease of notation, moving forward we remove the dependency of $\m{\Phi}_{0}^{k}(\m{x}_0)=\m{\Phi}_{0}^{k}$ on $\m{x}_0$. 
\begin{myrem}\label{rmk:Phi}
Notice that $\m{\Phi}_{0}^{k}$ represents the derivative of~\eqref{eq:disc_NDAE-ODE} with respect to $\m{x}_0$ for $k \in \{0\;,1\;,\;\ldots\;,\mr{N}-1\}$. This being said, the transition matrix $\m{\Phi}_{0}^{k}$ requires the knowledge of $\m{x}_k$ for all $k$. As such, we can apply the chain rule to evaluate $\m{\Phi}_{0}^{k}$ for any time-index $k$ as 
\begin{equation}\label{eq:Phi}
		\m{\Phi}_{0}^{k}
		= \m{\Phi}_{k-1}^{k}  \m{\Phi}_{k-2}^{k-1}
		\; \ldots \; 	\m{\Phi}_{0}^{1}\m{\Phi}_{0}^{0} 
		= \prod^{i=k}_{1}\m{\Phi}^{i}_{i-1}.
\end{equation}
\end{myrem}

The subsequent assumption is introduced to ensure that the system's response to perturbations or uncertainties remains bounded along the system trajectory. This is equivalent to considering a bound on the Jacobian of the nonlinear dynamical system, that is, $\norm{\m{J}(\tilde{\m{f}}(\cdot))} = \norm{\frac{\partial \tilde{\m{f}}(\cdot)}{\partial \m{x}_k}} < \infty$; see~\cite{Nugroho2022b}.
\begin{asmp}\label{assump:bounded}
Consider the variational system~\eqref{eq:DiscVarState}, the state-transition matrix $\m{\Phi}_{0}^{k}$ is bounded, i.e., 
$
\sup \left \{\;\| \m{\Phi}_{0}^{k}\|: k \in \mathbb{N}\; \right\}<\infty .
$
\end{asmp}

It follows that for every $\m{\delta}\m{x}_0 \in \Rn{n}$ there exists a unique solution for $\m{\delta}\m{x}_k \; \forall \; k \in \{0\;,1\;,\;\ldots\;,\mr{N}-1\}$. Assumption~\ref{assump:bounded} is a direct implication of the regularity condition of the nonlinear system~\eqref{eq:semi_NDAE_rep} mentioned in Assumption~\ref{assump:index}. To that end, the maximal LE in finite-time starting from initial state $\m{x}_{0} \in \Rn{n}$ can be defined as
\begin{equation}\label{eq:maxLE}
	{\lambda}_{\mr{MLE}} := \sup \left \{\;  \lim_{k\rightarrow \mr{N-1}}
	\frac{1}{k}\log \left( \dfrac{\norm{\m{\delta}\m{x}_{k}}}{\norm{\m{\delta}\m{x}_0}}\right)\;\right\} ,
\end{equation}
where ${\lambda}_{\mr{MLE}}$ is the largest eigenvalue of time-evolution of variational state-vector $\m{\delta}\m{x}_{k}$ with respect to the initial perturbation $\m{\delta}\m{x}_0$, i.e., the spectral norm of Cauchy-Green tensor defined as~\eqref{eq:Cauchy} under the action of the logarithm function; see~\cite{Ershov1998}.
\begin{equation}\label{eq:Cauchy}
	\m{\Xi} := {\m{\Phi}_{0}^{k}}^{\top}{\m{\Phi}_{0}^{k}};
\end{equation}
where matrix $\m{\Xi} \in \Rn{n\times n}$ represents the deformation of the initial perturbation vector $\m{\delta}\m{x}_0$ along the trajectory of the system. The sign of the MLE ${\lambda}_{\mr{MLE}}$ indicates the exponential divergence or convergence about a state orbit; see Fig.~\ref{fig:LyapExp}. The exponent is related to the average expansion and contraction along the directions of the state-space, such that, for any initial perturbation $\m{\delta}\m{x}_0$ the stability of the system can be defined according to the following
\begin{itemize}[itemindent=*]
\item Unstable if ${\lambda}_{\mr{MLE}} > 0$, i.e., system is chaotic and diverges.
\item Stable if ${\lambda}_{\mr{MLE}} \leq 0$, i.e., system trajectory is attracted to a stable or periodic orbit.
\end{itemize}

The maximal LE throughout the literature is typically applied to study a single aspect of stability, that being frequency, voltage, or rotor angle stability. In this paper, we present a method for stability assessment that considers all the above aspects simultaneously. As such, the following section introduces an LE-based stability method that entails computing the full spectrum of exponents associated with stability along different state-space trajectories.

\begin{figure}[t]
	\centering
	\subfloat{
		\begin{tikzpicture}[>=Stealth]
			\draw[thick] (0,0) 
			to[out=0,in=120] (2,-0.4) 
			to[out=-60,in=170] (4,0.1);
			
			\draw[thick,blue] (0,0.5)
			to[out=0,in=120] (2.1,-0.09)
			to[out=-60,in=160] (4,0.3);
			\draw[centered,->] (1,0.06) -- (1.1,0.06);
			
			\draw[dashed] (0,0) circle [radius=0.5];
			\draw[->,dashed] (0,0) -- (0,0.5);

			\draw[dashed] (4,0.1) circle [radius=0.2];
			\draw[->,dashed] (4,0.1) -- (4,0.3);
			
			\draw[thick,->] (1,0.5) -- (1.1,0.5);
			
			\draw (-0.8,0) node {$({a})$};
			\draw (0,-0.2) node {$\m{x}_0$};
			\draw (0,0.7) node {$\m{x}_0+\m{\delta}\m{x}_0$};
			
			\draw (4,-0.25) node {$\m{x}_t$};
			\draw (4,0.55) node {$\m{x}_t+\m{\delta}\m{x}_t$};
			\draw (5.5,0.275) node {$\m{\lambda_{MLE}} \leq 0$};
	\end{tikzpicture}}{}{}\vspace{-0.5cm}
	\subfloat{
		\begin{tikzpicture}[>=Stealth]
			\draw[thick] (0,0) 
			to[out=0,in=120] (2,-0.4) 
			to[out=-60,in=170] (4,0);
			
			\draw[thick,red] (0,0.5)
			to[out=0,in=120] (2.4,0.12)
			to[out=-60,in=-150] (2.8,0.16)
			to[out=35,in=-100] (4,0.8);
			
			\draw[centered,->] (1,0.06) -- (1.1,0.06);
			
			\draw[dashed] (0,0) circle [radius=0.5];
			\draw[->,dashed] (0,0) -- (0,0.5);

			\draw[dashed] (4,0) circle [radius=0.8];
			\draw[->,dashed] (4,0) -- (4,0.8);
			
			\draw[thick,->] (1,0.58) -- (1.1,0.59);
			
			\draw (-0.8,0) node {$({b})$};
			\draw (0,-0.2) node {$\m{x}_0$};
			\draw (0,0.75) node {$\m{x}_0+\m{\delta}\m{x}_0$};
			\draw (4,-0.2) node {$\m{x}_t$};
			\draw (4,1) node {$\m{x}_t+\m{\delta}\m{x}_t$};
			
			\draw (5.5,0.5) node {$\m{\lambda_{MLE}} > 0$};
	\end{tikzpicture}}\vspace{-0.5cm}
	\caption{Trajectory of nearby orbits starting from initial state $\m{x}_0$ and perturbed initial state $\m{x}_{0}+\m{\delta}\m{x}_{0}$: $(a)$ converging and $(b)$ diverging trajectories.}\label{fig:LyapExp}
\end{figure}
\subsection{Computation of Lyapunov spectrum of exponents}\label{subsec:LyapTheory}
The MLE is sufficient to determine the asymptotic stability of a dynamical system. In addition to asymptotic stability assessment, computing the spectrum of LEs provides a characterization of the behavior of system along all direction of the phase-space. That being said, the spectrum provides an identification of the direction or nodes of a system that are sensitive to a perturbation in initial conditions, i.e., it allows for the understanding of uncertainty propagation along the dynamic trajectories of a system. The spectrum of LEs can be computed as follows
\begin{equation}\label{eq:specLyapExp}
		\m{\Lambda}
		:=  \mr{spec}\left\{ \lim_{k\rightarrow \mr{N-1}} \frac{1}{k}\log  \left(\norm{{\m{\Phi}_{0}^{k}}}\right)\right\},
\end{equation}
where $\m{\Lambda} \in \Rn{n \times n}$ is a diagonal matrix representing the spectrum of LEs. Note that for a finite dimensional vector space the spectrum coincides with the eigenvalues of the matrix. This indicates that the spectrum of LEs are exactly the eigenvalues corresponding to the covariant Lyapunov vectors. For any regular dynamical system, these vectors represent the local decomposition of the phase space~\cite{Ginelli2007}. The local covariant vectors along with the exponents inform us on the local stability behavior along the phase trajectory of the system~\cite{Bosetti2018}.
\begin{myrem}\label{rmk:regular}
	The assumption pertaining to the regularity of the dynamical system~\eqref{eq:semi_NDAE_rep} is mild. The existence of the full spectrum of LEs is well-established as a result of the multiplicative ergodic theorem (Oseledets Theorem) proven in the 1960s; see~\cite{Barreira1998}. 
\end{myrem}
The existence of the full spectrum of LEs guarantees regularity based on Oseledets' Theorem. The reason is that computing LEs requires the computation of well-defined Jacobian matrices of the nonlinear system. These Jacobian matrices correspond to the variational mapping function in~\eqref{eq:DiscVarState}, whose existence is ensured by the properties guaranteed by Oseledets' Theorem. Therefore, regularity, defined by the existence and smoothness of such Jacobian matrices, is implied. It follows that for general discrete-time systems regularity is provided as a result of the aforementioned theorem~\cite{Manneville1990, Frank2018,Tranninger2020, Martini2022}. This result eliminates the technical challenges of computing the LEs while requiring to verification of the regularity assumptions.

 Computation of the spectrum of LEs is well-established~\cite{Ershov1998,Dieci2008}.
  However, generally the fundamental mapping function $\m{\Phi}_{0}^{k}$ for the variational system~\eqref{eq:DiscVarState} is ill-conditioned. Therefore, the error in computing the LEs increases and thus the solution tends to converge in the direction of the MLE~\cite{Masarati2015,Frank2018, Balcerzak2020}. To avoid such problems when practically computing the spectrum of exponents, one can rely on re-orthogonalization of the local directions of the mapping function $\m{\Phi}_{0}^{k}$ along the system trajectory. This results in a transformation of $\m{\Phi}_{0}^{k}$ to an upper triangular matrix. The existence of such triangular matrix for a regular dynamical system is a direct result of Perron's lemma~\cite[Lemma 1.3.3]{Barreira1998}.
 
There are two orthogonal factorization classes for estimating the LEs from the triangular matrix form (i) continuous QR- and (ii) discrete QR-method. We compute the spectrum of LEs by utilizing the discrete QR-method;~see~\cite{Dieci2008}. The matrix $\m{\Phi}^{k}_{0}$ and its triangular factor $\m{R}^{k}_{0}$ are directly evaluated by a re-orthogonalization integration via discrete-QR orthonormalization. For $k \in \{0,\;1,\ldots,\;\mr{N}-1\}$ the discrete-QR factorization of transition matrix $\m{\Phi}^{k}_{0}$~\eqref{eq:Phi} can be written as
\begin{equation}\label{eq:QRFactorization}
	\m{\Phi}^{k}_{0} = \m{Q}_{k}\m{R} ^{k}_{0}, \quad \m{\Phi}^{0}_{0} = \m{Q}_{0}\m{R} ^{0}_{0} = \m{I}_n,
\end{equation}
where $\m{Q}_{k} \in \Rn{n\times n}$ is an orthogonal matrix and $\m{R} ^{k}_{0}\in \Rn{n \times n}$ is a positive upper triangular matrix. Due to the uniqueness of the solution of the QR-factorization, as a result of regularity, the triangular matrix $\m{R} ^{k}_{0}$ can be written as
 \begin{equation}\label{eq:QR-R}
	\m{R}_{0}^{k} = \m{R}_{k-1}^{k}  \m{R}_{k-2}^{k-1}
	\ldots 	\m{R}_{0}^{1}\m{R}_{0}^{0} = \prod^{i=k}_{1}\m{R}^{i}_{i-1}.
\end{equation}

{In order to attenuate} the ill-conditioned matrix computations from~\eqref{eq:specLyapExp}, the spectrum of LEs can be expressed as
\begin{equation}\label{eq:specLyapExpQR}
	\small
		\m{\Lambda}
		:=  \mr{spec}\left\{ \lim_{k\rightarrow \mr{N-1}} \frac{1}{k}\log  \abs{\prod^{i=k}_{1}\m{R}^{i}_{i-1}} \right\}.
\end{equation}
\begin{myrem}\label{rmk:QR}
Notice that the diagonal elements of $\m{R}_{0}^{k}$ are only required for the computation of the spectrum of LEs. It follows that to compute LEs, we evaluate  $\abs{\m{R}_{0}^{k}}$ instead of $\norm{\m{R}_{0}^{k}}$;~see~\cite{Dieci2011}.
\end{myrem}
The discrete QR-factorization for the variational transition matrix~\eqref{eq:QRFactorization} can be obtained according to the methods developed in~\cite{Dieci2008}. We summarize the implementation of the discrete-QR method in the following algorithm.
\begin{algorithm}[h]
	\caption{Discrete-QR Factorization Algorithm}\label{algorithm:DQR}
	\DontPrintSemicolon
	\textbf{Input:} Initial factorization $\m{\Phi}^{k}_{0} = \m{Q}_{k}\m{R} ^{k}_{0}$\;
	\textbf{Output:} QR factorization $\; \forall \; k \; \in \; \{1,2,\ldots,\mr{N}-1\}$\;
	\textbf{Initialize:} $\m{Q}_0$, $\m{R}^{0}_{0}$\;
	\ForAll {$ k \; \in \; \{1,2,\ldots,\mr{N}-1\}$}{
		\textbf{compute:} $\m{\Phi}_{0}^{k}(\m{x}_0)$ $\m{\Phi}_{0}^{k}=\left(\m{I}_{n}  +\tfrac{\partial\tilde{\m{f}}( \m{x}_{k},\m{x}_{k-1})}{\partial\m{x}_{k-1}}\right)\tfrac{\partial\m{x}_{k-1}}{\partial\m{x}_{0}}, \quad\m{\Phi}_{0}^{0}(\m{x}_0)=\m{I}_{n}$\;
		\textbf{compute QR factorization:}	$\m{\Phi}^{k}_{0}\m{Q}_{k-1} = \m{Q}_{k}\m{R} ^{k}_{k-1}$\;
		\textbf{record and update:} $ \m{Q}_{k-1} \leftarrow  \m{Q}_{k}$\;
		\textbf{record and update:}$\m{R} ^{k-1}_{k-2} \leftarrow \m{R} ^{k}_{k-1}$\;
	}
\end{algorithm}
\setlength{\textfloatsep}{0pt}
Computing the LEs enables the formulation of a scalable method that assesses the stability of NL-DAE systems against perturbed load injections from renewables and can potentially inform power system operators about the impact of renewable injections on overall system stability.
\section{Stable node identification and optimal renewables allocation}\label{sec:main}
The previous section introduces Lyapunov's method for the stability analysis of dynamical systems. In this section, we present a framework based on the aforementioned stability quantification method for identifying stable nodes in NL-DAE  power systems. The proposed framework enables the allocation of uncertain and intermittent loads from RERs while considering the impact on the overall power network stability.
\subsection{Stable node identification}\label{subsec:StableNodeIdent}
The high penetration of renewables, along with their dynamic behavior and lower system inertia, impose a challenge on the transient stability of the system. 
Typically, methods that assess stability based on the LEs of a power system consider model-free computations. This implies that the LEs are computed based on the voltage stability of the buses. However, although some methods mentioned in this paper consider rotor angle stability, they rely on state estimation. This indeed results in computational errors when estimating the LEs that measure the infinitesimal trajectory of state perturbations. In this paper, having provided a model-based approach for computing the LEs of a NL-DAE power system representation, we present a framework to evaluate node stability by considering a more comprehensive approach.  

With that in mind, we assess the stability of the nodes, i.e., buses, of a power system based on the several stability criteria: (i) voltage stability, (ii) rotor angle stability, and (iii) frequency stability. The stable node identification framework is based on computing the LEs while considering dynamic system response after a perturbed RER load. To study the stability of a power network after a load disturbance resulting from such intermittent and uncertain RER injections, we introduce a quantitative measure equivalent to computing the spectrum of LEs of the power system. The following proposition establishes this expression.
\begin{myprs}\label{prs:param_Lyapexp}
	The {parameterized} tensor matrix~\eqref{eq:Cauchy} representing the state-deformation along the trajectory of the nonlinear discrete-time power system~\eqref{eq:disc_NDAE-ODE} can be expressed as follows
	\begin{equation}\label{eq:paramLyap}
		\tilde{\m{\Xi}}(\m\gamma) := \sum_{j=1}^{n} \gamma_j 
		\left(\sum_{i=0}^{k}\left(\m{\varphi}^{i}_{0}\right)^\top \m{\varphi}^{i}_{0} \right) \; \in \; \Rn{n \times n},
	\end{equation}	
where $\m{\varphi}^{i}_{0}$ represents the column vectors of matrix $\m{\Phi}_{0}^{k}$. The parameterization $\gamma_j$ determines the states that are required for the stability assessment. That is, if $\gamma_j= 1$ then the state is considered for LE computations, else $\gamma_j= 0$. The {parameterized} vector $\m \gamma  \in \Rn{n}$ that represents the selected states is {defined as} $\nocol{\gamma}{j}{1}{n}$. 
\end{myprs} 
\begin{proof}\label{proof:param_LyapExp}
	{From \eqref{eq:Cauchy}, being multiplied with parameterization vector $\m{\gamma}$, it follows that}	
	\begin{subequations}\label{proof:prop1-cauchy}
		\begin{align}
			\tilde{\m{\Xi}}(\m\gamma)  &= {\m{\Phi}^{k}_{0}}^{\top} \bmat{\m{I} \otimes \m{\gamma}}^\top \bmat{\m{I} \otimes \m{\gamma} }	\m{\Phi}^{k}_{0},\label{eq:proof:prop1-1} \\
			&= \sum_{i=0}^{k} \left({\m{\varphi}^{i}_{0}} \right)^\top \m{\gamma}^{2} \m{\varphi}^{i}_{0}
			= \sum_{i=0}^{k}\sum_{j=1}^{n}\gamma_j \left(\m{\varphi}^{i}_{0}\right)^\top \m{\varphi}^{i}_{0},\label{eq:proof:prop1-3}
		\end{align}
	\end{subequations}	
	where~\eqref{eq:proof:prop1-1} is a result of applying the dot product and the  Kronecker product. {Furthermore,~\eqref{eq:proof:prop1-3}} holds true since $\m \gamma^2 = \m \gamma$. The proof is complete since \eqref{eq:proof:prop1-3} is equivalent to \eqref{eq:paramLyap}.
\end{proof}
 We note that throughout this section we consider deformation matrix~\eqref{eq:Cauchy} instead of~\eqref{eq:specLyapExpQR} for developing the overall stability quantification measures. {However, in the case studies we implement the algorithms} by computing the spectrum of LEs based on the QR-factorization mentioned above. The reason is to simplify the proofs related to how we develop the stability measures and their relation to LEs. Such that under the QR factorized matrices, an unnecessary layer of complexity will be added to the proofs. 

The above {parameterized} deformation matrix~\eqref{eq:paramLyap} represents the deformation along the trajectory for the selected states. Such a parameterization allows for choosing, at each bus, the states that contribute to the stability computation. If a bus is a generator bus then we can choose from the differential states~\eqref{diffstate} and algebraic state vector~\eqref{algstate}. If it is a load bus we can choose from the algebraic state vector~\eqref{algstate}. {This allows, for a given bus, to choose a subset of states for the stability quantification. However when selecting more than one state, we obtain several LEs at each of the buses.} To present a singular LE for a bus, we first introduce some properties for LEs that are necessary for the unified quantification.  
\begin{mypro}\label{property:LyapProp} 
	LEs $\m{\lambda}$ exhibit the following properties
	\begin{align*}\vspace{-0.2cm}\label{eq:prop_LEs}
		(\mr{P}1.1) &\quad \m{\lambda}(\beta \m{A}) = \m{\lambda}( \m{A}) \; \forall \; \beta \in \mathbb{R} \backslash \{0\}, \\ \nonumber
		(\mr{P}1.2) &\quad \m{\lambda}(\m{A}_1+\m{A}_2) \leq \max\{\m \lambda(\m{A}_1),\m \lambda(\m{A}_2)\}, \nonumber
	\end{align*}
	where $\m{A}_1$ and $\m{A}_2$ are of dimension $\Rn{n \times n}$; see~\cite[Theorem 2.1.2]{Barreira1998}.
\end{mypro}
Based on the above properties, we now present the LE-based bus stability identification in NL-DAE models of power networks. The following proposition establishes such  stability quantification method.
\begin{myprs}\label{prs:stability-ident}
	Let $ i \; \in \; \mc{N}$ denote the bus index, such that $i_{\mc{G}} \; \in \; \mc{G}$ and $i_{\mc{L}} \; \in \; \mc{L} \cap \mc{R}$ represent the generator and load/renewable buses. {Then, the stability of a node in a power network, characterized by its LEs, can be expressed as follows}
\begin{equation}\label{eq:node_stable}
	\hspace{-0.05cm}\lambda_{i} \hspace{-0.05cm}:=\hspace{-0.05cm} \begin{cases}
			\lambda_{i_{\mc{G}}}\hspace{-0.05cm} = \hspace{-0.05cm}
			\underset{k\rightarrow \mr{N-1}}{ \lim} \left[\frac{1}{2k}\log \left(\tilde{\m{\Xi}}(\m\gamma)\right)\right],
			&\hspace{-0.2cm} \gamma_j\hspace{-0.05cm} \in\hspace{-0.05cm} \left\{ \delta_i, \omega_i,v_i \right\}_{i \; \in \; \mc{G}}, \\
			 \lambda_{i_{\mc{L}}}\hspace{-0.05cm} = \hspace{-0.05cm}
			\underset{k\rightarrow \mr{N-1}}{ \lim}  \left[\frac{1}{2k}\log \left(\tilde{\m{\Xi}}(\m \gamma)\right)\right],
			&\hspace{-0.2cm} \gamma_j \hspace{-0.05cm}\in \hspace{-0.05cm}\left\{v_i\right\}_{i \; \in\;  \mc{L}},
	\end{cases}
\end{equation}
where $j \in \Rn{n}$ represents the index for the differential and algebraic states defined by state vector $\m{x}_k$.
\end{myprs} 
\begin{proof}\label{proof:stability-ident}
	For the proof, we consider a generator bus $i\; \in \; \mc{G}$. Let $\tilde{\m{\Xi}}(\m\gamma)$ define the deformation matrix of the power grid. From Proposition~\ref{prs:param_Lyapexp}, we can write the following
	\begin{equation}\vspace{-0.1cm}\label{proof:prop2-1-stability}
		\tilde{\m{\Xi}}(\m\gamma)
		= \sum_{i=0}^{k}\sum_{j=1}^{n}\gamma_j \left(\m{\varphi}^{i}_{0}\right)^\top \m{\varphi}^{i}_{0}, \quad \gamma_j\hspace{-0.05cm} \in\hspace{-0.05cm} \left\{ \delta_i, \omega_i,v_i \right\}_{i \; \in \; \mc{G}},			
	\end{equation}	
	where applying the $\mr{log}$ function to the above matrix we obtain the following 
	\begin{equation}\label{proof:prop2-2-stability}
		\lambda_{i_{\mc{G}}}
		= \frac{1}{2k} \mr{log} \left(\sum_{i=0}^{k}\sum_{j=1}^{n}\gamma_j \left(\m{\varphi}^{i}_{0}\right)^\top \m{\varphi}^{i}_{0}\right).
	\end{equation}
	
	Taking the limit of the $\mr{log}$ and applying properties $(\mr{P}1.1)$ and $(\mr{P}1.2)$ we obtain the LEs of the power network buses for $i\; \in \; \mc{G}$. This holds true for load/renewables buses $i \; \in \; \mc{L} \cap \mc{R}$.
\end{proof}
The above proposition establishes the stability of a node in a power grid; it is based on the type of bus and thus the states for the stability assessment are represented by the parameterization $\gamma_j$. For a generator bus $i\in\mc{G}$, $\gamma_j$ represents rotor angle, frequency and voltage stability, whereas for a load/renewables bus $i\in\mc{L} \cap \mc{R}$, $\gamma_j$ represents voltage stability. 
\begin{myrem}\label{rmk:Lyap_properties}
	We note here that for generator buses, i.e., $i \in \mc{G}$, the state variables $ \left\{ \delta_i, \; \omega_i,\; v_i \right\}_{i \; \in \; \mc{G}}$ contribute to the computation of the nodes' LEs. {These exponents are bounded by the maximum exponents} computed from each of the variables, a result of Property $(\mr{P}1.2)$. That being said, the following upper bound holds true 
	\begin{equation}\label{eq:upperbound_Lyap}
		\lambda_{i_{\mc{G}}}  \leq  \max\{\lambda_{i_{\mc{G}}} (\left\{ \delta_i\right\}_{i \in \mc{G}}),\;
		\lambda_{i_{\mc{G}}} ( \left\{ \omega_i\right\}_{i \in \mc{G}}),\;
		\lambda_{i_{\mc{G}}} (\left\{ v_i\right\}_{i \in \mc{G}})\}.
	\end{equation}
\end{myrem}

We approach quantifying stability of a NL-DAE power network from a dynamical systems' perspective without the need for a computationally demanding approach. This approach allows us to consider the joint modeling of the interactions between frequency, voltage, rotor angle, and other states of the system. {The joint modeling enables us to perform stability quantification across different time scales within a single framework. In the literature, ODE models are considered to simplify stability studies by decoupling the stability analysis of the differential and algebraic states. This decoupling does not fully capture the interdependence between different types of stability and often overlooks important interactions, particularly in systems with RERs. By considering a NL-DAE representation instead of a simplified ODE formation that decouples the algebraic constraints we can incorporate both the differential dynamics (such as those for frequency and rotor angle stability) and the algebraic constraints (such as those for voltage stability).} {It is worth noting that other stability aspects, such as inverter angle stability, can also be considered when modeling more comprehensive models with PV plants.}

Having presented the above node stability quantification method, we now rank the nodes according to their corresponding LE value. As discussed earlier, a stable state trajectory has a negative LE. As such, let $\mr{S}_i$ define the stability index of a node in a power system network. The stability index of a node is equivalent to the index of the LE from the set of ordered exponents in ascending order and can be expressed as
\begin{equation}\label{eq:stability_index}
	\mr{S}_i := \mr{index}(\lambda_i), \quad \lambda_{i} \in \tilde{\m \Lambda},
\end{equation}
where $\tilde{\m{\Lambda}} = \mr{col}\left\{\lambda_{i}\right\}_{i \in \mc{N}}$ represents a column vector of the ordered LEs $\lambda_i$ in ascending order. The next section introduces the overall stability metric that quantifies the impact of uncertain RER load perturbation on the power grid as a whole. This metric indicates the uncertainty propagation onto the power grid from an perturbed RER load injection that is being applied to a single node within the network.
\subsection{Quantifying stability against renewable perturbations and optimally allocating their locations}\label{subsec:OptRERAlloc}
The method for assessing the stability of buses within a network relies on estimating the LEs and then ranking the stability of the buses according to the values of the exponents. That is, buses with larger negative LEs are considered more stable and thus are given a higher stability index $\mr{S}_i$. The stability index that is based on the rank, under the context of renewables allocation, does not assess the impact of such renewable load injection, along with its propagation, on the overall stability of the network.  
 
With that in mind, {we approach allocating a RER by considering} a LE-based stability measure that enables studying the impact of perturbed renewable injections on the overall stability of the power network. Meaning that, we want to quantify the impact of an uncertain renewable injection on the overall stability of the power system and then choose the nodes that {least impact overall system stability}. Before posing the RER allocation problem within power networks that are described by a complete NL-DAE representation, we introduce the following stability measure that is equivalent to computing the sum of LEs along all the nodes of the power network. The LE-based measure quantifies the uncertainty that is propagated with the power grid from a perturbed load injection at a network node.
\begin{theorem}\label{theo:logdet}
	The $\mr{log\,det}$ of matrix $\tilde{\m{\Xi}}(\m\gamma)$ quantifies the overall stability of a system. As such, the $\mr{log\,det}(\tilde{\m{\Xi}}(\m\gamma))$ is said to be equivalent the Lyapunov spectrum of exponents~\eqref{eq:specLyapExpQR} according to the following relation
	\begin{equation}\label{eq:Lyap_connection}
	\mr{log\,det}(\tilde{\m{\Xi}}(\m\gamma)) \equiv \beta \sum_{i=1}^{N}\lambda_{i}, \quad \forall \; i\in \mc{N},
	\end{equation}
	where $\beta$ is a constant equivalent to $	\tfrac{1}{2\mr{N}}$. 
	The constant $\beta$ is due to the definition of LEs being computed along the system trajectory.
\end{theorem}
\begin{proof}\label{proof:log-det}
	From~\eqref{eq:node_stable} observe that the LEs are the eigenvalues of the $\mr{log}$ of matrix $\tilde{\m{\Xi}}(\mc{S},\m\gamma)$. This matrix is positive semi-definite, whereby its determinant can be computed as
	\begin{align} \vspace{-0.1cm}\label{eq:log-det1}
		\mr{det}(\tilde{\m{\Xi}}(\m\gamma))
		&= \left(\prod_{i=1}^{N}\hat{\lambda}_i\right),	\vspace{-0.3cm}
	\end{align}
	where $\hat{\lambda}_i$ is the $i$-$th$ eigenvalue of matrix $\tilde{\m{\Xi}}(\m\gamma))$. Now taking the $\mr{log}$ of the eigenvalues $\hat{\lambda}_i$, we obtain according to~\eqref{eq:node_stable} the $i$-$th$ LE. As such, applying $\mr{log}$ to~\eqref{eq:log-det1} we obtain the following
	\begin{align}\label{eq:log-det2}
		\mr{log\,det}(\tilde{\m{\Xi}}(\m\gamma))
		&= \beta\left(\sum_{i=1}^{N}{\lambda}_i\right),	
	\end{align}
	this concludes the proof.
\end{proof}

The above theorem presents a method for stability quantification after an uncertain perturbation is applied at a system node; it suggests that computing~\eqref{eq:Lyap_connection} is equivalent to computing the overall stability of all the nodes in a power network after a perturbed renewable load is applied to a single node. As such, this metric enables the quantification of the impact on stability, through computing the spectrum of LEs, resulting from the allocation of uncertain RER at a node within a power network.

In order to identify the set of nodes that result in minimal uncertainty propagation, i.e., {nodes that are indicative of stability under a renewable load injection,} we are required to solve a combinatorial optimization problem. With that in mind, we now pose the optimal RER allocation problem based on the overall stability impact, computed based on~\eqref{eq:Lyap_connection}, {of a renewable injection at the power network buses.} This combinatorial problem is computationally expensive; it increases in complexity for large networks. As such, we pose the above problem as a set optimization problem for which we provide a computationally efficient algorithm to solve. With that in mind, let $\mc{S}$ denote the set of RERs to be allocated within an existing power network. Let the maximum number of RERs be $N$. Note that the maximum number of RERs is limited to the number of buses within a power network. 

\begin{algorithm}[t]
	\small
	\caption{Framework for identifying optimal RER allocation within a power network}\label{alg:algorithm_allocation}
	\DontPrintSemicolon
	\textbf{input:} Network parameters, $s$, $i_{\mc{G}}$, $i_{\mc{N}}$, $\mc{V}$, $\mc{N}$ \;
	\textbf{initialize:} $({\tilde{P}}_\mr{R_i}^0,{\tilde{Q}}_\mr{R_i}^0) \leftarrow (\m{0},\m{0})$, $i \leftarrow 1$, $k \leftarrow 1$ \;
	\For{$i  = 1$ \textbf{to} $N$}{
		\textbf{assign:} $({\tilde{P}}_\mr{R_i}^0,{\tilde{Q}}_\mr{R_i}^0) \leftarrow (1+\tfrac{\beta}{100})(\mr{P}_\mr{R_i}^0,\mr{Q}_\mr{R_i}^0 )$\; 
		\For{$k = \;1 $ \textbf{to} $\mr{N}-1$}{
			\textbf{simulate:}~\eqref{eq:disc_ssm_NDAE} by implementing Algorithm~\ref{algorithm:NR}\;
			\textbf{compute:} $\tilde{\m{\Xi}}_i(\m\gamma)$ by implementing Algorithm~\ref{algorithm:DQR}\;
		}
	}
	\textbf{initialize:} $\mc{S}\leftarrow \emptyset$, $i \leftarrow 1$, ${\alpha_i} \leftarrow {0} \; \forall \; i \; \in \; \mc{N}$ \;
	\For{$i \leq s$}{
		\textbf{set:} $\mc{L}{(\mc{S})} = \mr{log\,det} \left(	\tilde{\m{\Xi}}(\mc{S},\m\gamma)\right) = \mr{log\,det} \left( \sum_{i=1}^{N} \alpha_i 
		\tilde{\m{\Xi}}_i(\m\gamma)\right)$\;
		\textbf{compute:} $\mc{G}_i = \mc{L}(\mc{S}\cup \{a\})-\mc{L}(\mc{S})$, $\forall a\in \mc{V}\setminus \mc{S}$ \;
		\textbf{assign:} $\mc{S}\leftarrow \mc{S} \cup \left\{\mathrm{arg\,max}_{a\in \mc{V}\setminus \mc{S}}\,\mc{G}_i \right\}$ \;
		\textbf{update:} $i \leftarrow i + 1$\;
	}
	\textbf{output:} $\mc{S}^{*}$\;
\end{algorithm}
\setlength{\textfloatsep}{0pt}

The optimal RER allocation problem $\mr{\textbf{P1}}$ for the NL-DAE power system~\eqref{eq:disc_NDAE-ODE} can be written as a set optimization problem by defining 
the \textit{set function} $\mc{L}{(\mc{S})}: 2^{\mc{V}}\rightarrow \mbb{R}$ as the  $\mr{log\,det}$ of {parameterized} matrix $\tilde{\m{\Xi}}(\mc{S},\m\gamma)$, where $\mc{V} := \{ i\in\mbb{N}\,\;|\,\;\;0 < i \leq N\}$.
\begin{subequations}\label{eq:Opt_RER_Alloc_Prop}
	\begin{align}
		(\mr{\textbf{P1}})\;\; \maximize_\mc{{\m{S}}} \;\;\;
			& \mc{L}{(\mc{S})}:= \mr{log\,det} \left(	\tilde{\m{\Xi}}(\mc{S},\m\gamma)\right),\label{eq:Opt_RER_Alloc_Prop_1}\\
		\subjectto \;\;\;&  \abs{\mc{S}} = s, \;\; \mc{S}\subseteq\mc{V},
			\label{eq:Opt_RER_Alloc_Prop_2}
	\end{align}
\end{subequations}
where $\tilde{\m{\Xi}}(\mc{S},\m\gamma)\; \in \; \Rn{n \times n}$ represents the state-trajectory deformation matrix resulting from allocating an uncertain renewable load injection at a node defined by set $\mc{{S}}$. As such, $\tilde{\m{\Xi}}(\mc{S},\m\gamma)$ can be defined as
\begin{equation}\label{eq:deform_mat}
	\tilde{\m{\Xi}}(\mc{S},\m\gamma) := \sum_{i=1}^{N} \alpha_i 
	\tilde{\m{\Xi}}_i(\m\gamma), \quad \alpha_i \;\in\; \{0,1\}^{\mc{S}},
\end{equation}
where $\alpha_i$ is equivalent to zero if set $\mc{S}$ has a zero at the $i$-$th$ node index, and is one otherwise. This means that $\m{\alpha} := \{\alpha_i\}_{i\in\mc{N}}$ represents the allocation of RERs within the power network, where a value of one is given to a node with an RER injection. Matrix $\tilde{\m{\Xi}}_i(\m\gamma)$ denotes the deformation of state trajectories when an RER load injection is applied to a single node of index $i \in \mc{N}$. 
We note that based on Theorem~\ref{theo:logdet}, the optimal allocation problem $\mr{\textbf{P1}}$ is equivalent to computing the sum of the LEs across all the nodes after the application of a renewable load injection at a single perturbed node. As such, at each iteration of the allocation problem $\mr{\textbf{P1}}$, we solve the maximum value for $\mc{L}(\mc{S})$ by choosing the node that has the largest sum of LEs. This allocation problem will result in an optimal set $\mc{S}^{*}$ representing the ranked nodes from most critical to most stable. This is due to the fact that we require at each node to have a negative LE. The solution provides the operator the nodes that are indicative of stability; it informs the operator about which nodes in the system, when allocated a perturbed renewable load, lead to an unstable or stable system. Accordingly, to obtain the most stable node, we set $|\mc{S}| = N$ and choose the most stable node to be the last index of set $\mc{S}^{*}$.



Notice that the optimal problem $\mr{\textbf{P1}}$ is a combinatorial set optimization problem, and thus it is computationally expensive. To avoid the complexity of utilizing global solvers for solving the combinatorial problem, we exploit the \textit{submodular}\footnote{ A set function $\mc{L}: 2^{\mc{V}}\rightarrow \mbb{R}$ is said to be submodular if and only if for any $\mc{A},\mc{B}\subseteq\mc{V}$ with $\mc{A}\subseteq\mc{B}$, it holds true that $\mc{L}(\mc{A}\cup\{s\}) - \mc{L}(\mc{A})\geq 	\mc{L}(\mc{B}\cup\{s\}) -  \mc{L}(\mc{B})$.} and monotone increasing\footnote{A set function $\mc{L} : 2^V \rightarrow \mathbb{R}$ is monotone increasing  if $\forall \; A,B \subseteq V$, it holds true that $A \subseteq B \rightarrow \mc{L}(A) \leq \mc{L}(B)$} properties of the $\mr{log\,det}$ set function that allows for solving $\mr{\textbf{P1}}$ using \textit{greedy} algorithms. The following proposition establishes the submodularity of $\mr{log\,det}$ set function. 

\begin{myprs}\label{prs:Trace-Logdet_prop}
The $\mr{log\,det} $ of~\eqref{eq:deform_mat} denoted by set function $\mc{L}(\mc{S}):2^{\mc{V}}\rightarrow\hspace{-0.05cm}\mbb{R}$ as follows
\begin{equation}\label{eq:logdet_submodular} 
		\mc{L}(\mc{S}) =\mr{log\,det} \left(	\tilde{\m{\Xi}}(\mc{S},\m\gamma)\right),
\end{equation} 
is, for $\mc{S}\subseteq\mc{V}$, submodular and monotone increasing.
\end{myprs}
\begin{proof}\label{proof:submodular-log-det2}
	Let $\mc{L}_s: 2^{V-\{s\}} \rightarrow \mathbb{R}$ denote a derived set function is defined as follows
	\begin{align*}
		\mc{L}_s(\mc{S}) & =\mr{log}\,\mr{det} \left(\tilde{\m{\Xi}}\left({\mc{S} \cup\{s\}}\right)\right)-\mr{log}\,\mr{det}\left(\tilde{\m{\Xi}}(\mc{S})\right),\\
		&=\mr{log}\,\mr{det}\left(\tilde{\m{\Xi}}(\mc{S})+\tilde{\m{\Xi}}(\{s\})\right)-\mr{log}\,\mr{det}\left(\tilde{\m{\Xi}}(\mc{S})\right).
	\end{align*}
	We first show $\mc{L}(\mc{S})$ that is monotone decreasing for any $s \in V$. Let $\mc{A} \subseteq \mc{B}\subseteq \mc{V}-\{s\}$, and let $\tilde{\m{\Xi}}(\m{\alpha}) =\tilde{\m{\Xi}}(\mc{A}) +\m{\alpha}\left(\tilde{\m{\Xi}}(\mc{B})-\tilde{\m{\Xi}}(\mc{A})\right)$ for $\m{\alpha}\in[0,1]$. Then for 
	\begin{align*}
		\tilde{\mc{L}_s}(\tilde{\m{\Xi}}(\m{\alpha}))=\mr{log}\,\mr{det}\left(\tilde{\m{\Xi}}(\m{\alpha})+\tilde{\m{\Xi}}(\mc{S})\right)-\mr{log}\,\mr{det}\left(\tilde{\m{\Xi}}(\m{\alpha})\right),
	\end{align*}
	we obtain the following
		\begin{align*}
				\frac{\mr{d}}{\mr{d} \m{\alpha}} &\tilde{\mc{L}_s}(\tilde{\m{\Xi}}(\m{\alpha}))\\
				&\quad= 
				\mr{trace}\Big[
				\left(\left(\tilde{\m{\Xi}}(\m{\alpha})+\tilde{\m{\Xi}}(\mc{S})\right)^{-1}-\;\;\tilde{\m{\Xi}}(\m{\alpha})^{-1}\right)\\
				&\quad\quad\quad\quad\quad
				\left(\tilde{\m{\Xi}}(\mc{B})-\tilde{\m{\Xi}}(\mc{A})\right) \Big] \leq 0.
			\end{align*}
	Such that $$	\left(\hspace{-0.1cm}\left(\tilde{\m{\Xi}}(\m{\alpha})+\tilde{\m{\Xi}}(\mc{S})\right)^{-1}-\;\;\tilde{\m{\Xi}}(\m{\alpha})^{-1}\right)^{-1} \hspace{-0.2cm}\preceq 0,$$ 
	and $\left(\tilde{\m{\Xi}}(\mc{B})-\tilde{\m{\Xi}}(\mc{A})\right)\succeq 0,$ then the above inequality holds. Thus, we have $\mc{L}_s$ is monotone decreasing, and $\mc{L}(\mc{S})$ is submodular. 
	On such note, the above proposition holds true. 
\end{proof}
 The aforementioned proposition establishes the submodularity of the optimal allocation problem $\mr{\textbf{P1}}$, this allows the use of greedy algorithms to solve the computationally exhaustive problem in a more efficient method. The greedy algorithm~\cite[Algorithm 1]{Kazma2023f} is utilized to solve the aforementioned submodular optimization problem. The optimality guarantee of the greedy algorithm for submodular set optimization is detailed in Appendix~\ref{apndx:greedy}. We note that this optimality guarantee reaches $99\%$ optimality when solving an allocation problem that has a submodular objective function. 
 
 \begin{figure*}[t]
 	\vspace{-0.3cm}	\centering
 	\subfloat{\includegraphics[keepaspectratio=true,scale=0.52]{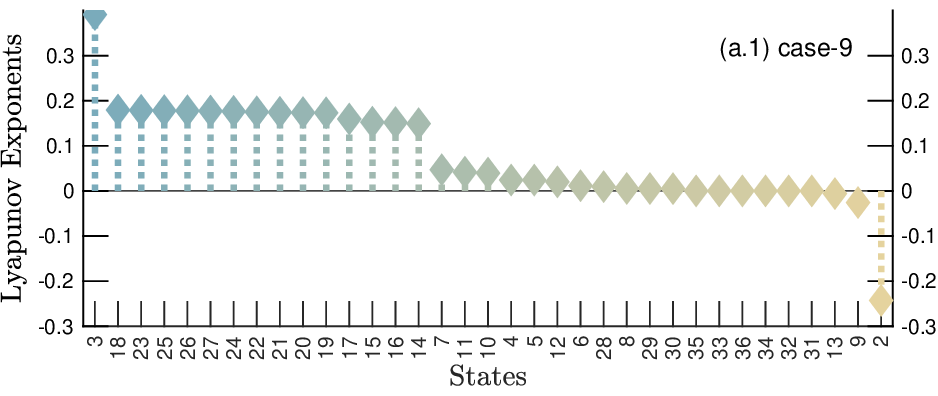}}\hspace{0.35cm}
 	\vspace{-0.35cm}
 	\subfloat{\includegraphics[keepaspectratio=true,scale=0.52]{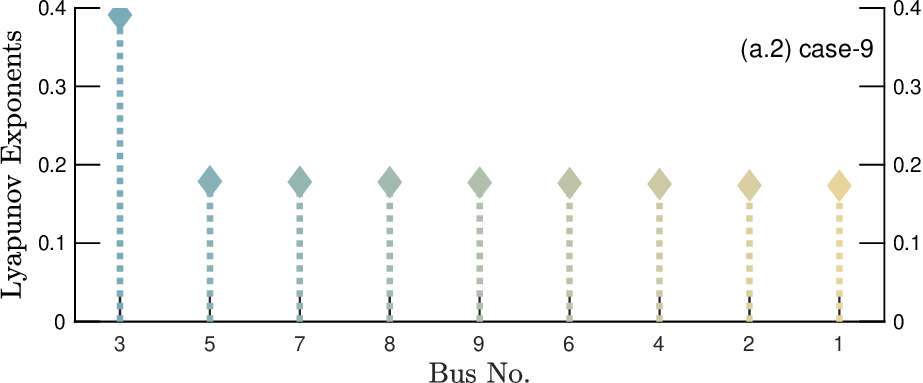}}{}{}
 	\subfloat{\includegraphics[keepaspectratio=true,scale=0.52]{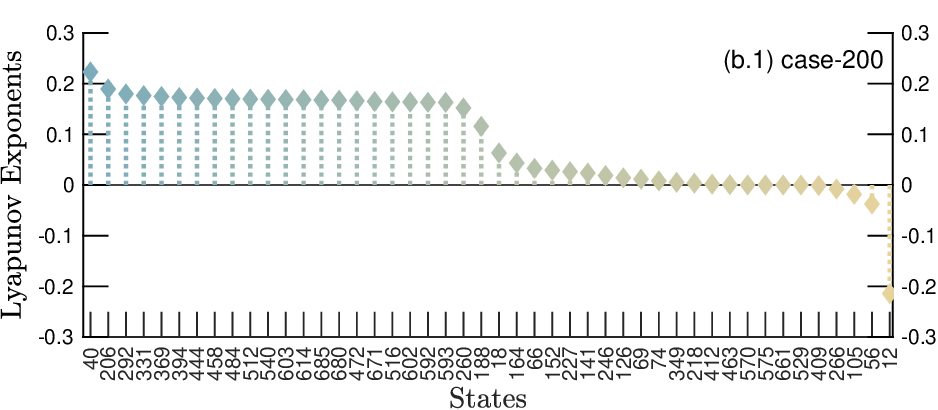}}\hspace{0.35cm}
 	\subfloat{\includegraphics[keepaspectratio=true,scale=0.52]{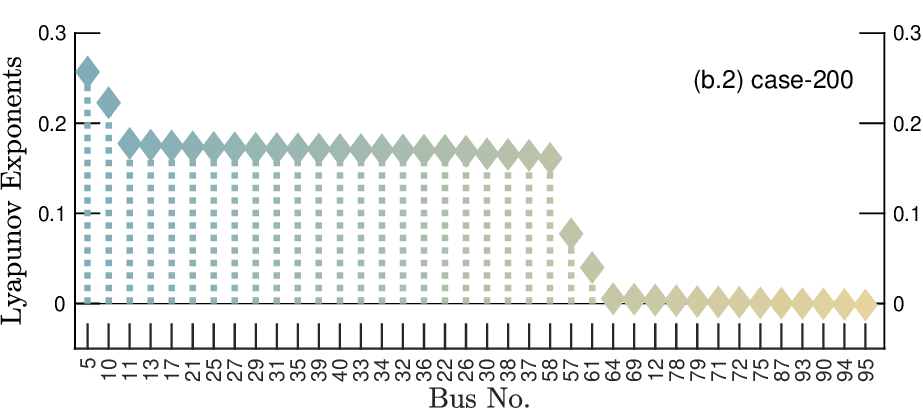}}{}{}
 	\vspace{-0.1cm}
 	\centering 
 	\caption{Lyapunov spectrum of exponents for (a.1 and a.2) $\mr{case}$-$\mr{9}$ (9-buses) and (b.1 and b.2) $\mr{case}$-$\mr{200}$ (200-buses). Column one depicts the spectrum of LEs computed for all the system states. Column two depicts the LEs for each bus within the network, i.e., the stability index.}\label{fig:LyaExp-9}
 	\vspace{-0.6cm}
 \end{figure*}
 
The proposed problem $\mr{\textbf{P1}}$ can be solved efficiently according to the greedy algorithm. The framework for identifying optimal stable nodes and uncertain renewable injection locations is outlined in Algorithm~\ref{alg:algorithm_allocation} which embeds the greedy algorithm to obtain $\mc{S}^*$.
\section{Numerical case study}\label{sec:casestudy}
In this section, we demonstrate the proposed stable node identification and RER allocation framework on standard power systems. The goal here is to identify stable nodes that allow for attenuation of the uncertain and intermittent  loads from RER on the overall stability of a power network. With that in mind, we attempt to answer the following questions.
\begin{itemize}[leftmargin=*]
	\item $(\mathrm{Q}1)$ Does the stability index $\mr{S}_i$ depict different forms of system instability (frequency, voltage and rotor angle) in power systems?
	\item $(\mathrm{Q}2)$ Does solving the optimal RER allocation framework result in optimal stable nodes that offer the least overall stability implication on the power network?
	\item $(\mathrm{Q}3)$ Is the clearing time after inducing a disturbance from an RER load at an optimal node shorter than disturbances added at other nodes? This verifies the optimality of the allocation problem.
	\item $(\mathrm{Q}4)$ Is the overall stability measure numerically equivalent to the summation of LEs (see Theorem~\ref{theo:logdet})?
	\item $(\mathrm{Q}5)$ Can the stable node identification framework inform system operators on where to allocate RERs into the system?
\end{itemize}
\subsection{Implementation of the stability and allocation framework}\label{sec:OPP_Frame}
As mentioned in Section~\ref{subsec:NDAE}, we consider two power system models detailed in Appendix~\ref{apndx:power_model} and Appendix~\ref{apndx:power_model2}. That being said, we attempt to answer $\mathrm{Q}1$-$\mathrm{Q}4$ for the $4^{\text{th}}$-order power system model first, and then do the same for the $9^{\text{th}}$-order PV-integrated power grid model in Section~\ref{subsec:RER_allocation}. We first consider three standard power systems of contrasting size for the assessment of the proposed framework: (i) $\mr{case}$-$\mr{9}$: A $\mr{9}$-bus power system with $\mr{3}$ synchronous generators---Western System Coordinating Council (WSCC), (ii) $\mr{case}$-$\mr{39}$: A 39-bus network with 10 synchronous generators---New-England Power System, and (iii) $\mr{case}$-$\mr{200}$: A $\mr{200}$-bus power system with $\mr{49}$ synchronous generators---ACTIVSg200-Bus network “Illinois200 case.” The simulations and optimization problem are performed in MATLAB R2021b running on a MacBook Pro with an Apple M1 Pro chip, a 10-core CPU, and 16 GB of RAM.
 
The generator parameters are extracted from power system toolbox (PST)~\cite{Sauer2017} case file $\mr{data3m9b.m}$ for $\mr{case}$-$\mr{9}$, and $\mr{datane.m}$ for $\mr{case}$-$\mr{39}$. For $\mr{case}$-$\mr{200}$ the generator parameters are chosen based on the ranges provided in the PST toolbox. Regulation and chest time constants for the generators are chosen as $R_{\mr{D}i} = 0.2 \; \mr{Hz/pu}$ and $T_{\mr{CH}i} = 0.2 \; \mr{sec}$. The synchronous speed is set to $\omega_{0} = 120\pi \;  \mr{rad/sec}$ and a power base of $100 \; \mr{MVA}$ is considered for the power system. The steady-state initial conditions for the power system are generated from solutions of the power flow obtained from MATPOWER $\mr{runpf}$ function. 

The discretization constant for the step-size is $h = 0.1$ and simulation time-span $t = 30 \; \mr{sec}$. We simulate a renewable load disturbance by introducing at  $t > 0$ a RER load denoted as $(\mr{P}_\mr{R}^0,\mr{Q}_\mr{R}^0)$ using the NR method as described in Appendix~\ref{apndx:NR-method}. In the scope of this work, RERs are modeled as negative loads that are injected into the power network at a node in set $\mc{N}$. Uncertainty of an RER load is modeled by adding a load perturbation of magnitude $\beta$. The perturbed renewables load at a node is then computed as $(\mr{\tilde{P}}_\mr{R}^0,\mr{\tilde{Q}}_\mr{R}^0) = (1+\tfrac{\beta}{100})(\mr{P}_\mr{R}^0,\mr{Q}_\mr{R}^0 )$. To account for the uncertainty of the renewable load, we vary $\beta$ between $\{2\%, 20\%\}$. To verify accuracy of the proposed NL-DAE transformation readers are referred to~\cite{Kazma2023b}. Note that the power generator's~\eqref{eq:differential_dynamics} response to the transients induced by the uncertain renewable injections is automatically regulated by the governor response $\m{T}_\mr{r}$. 

\begin{figure*}[t]
	\centering
	\subfloat{\includegraphics[keepaspectratio=true,scale=0.6]{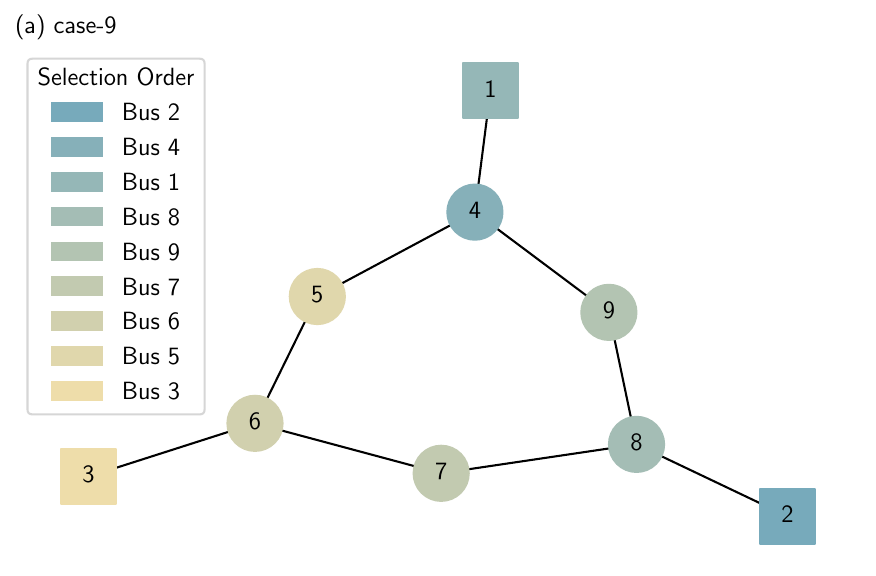}}
	\subfloat{\includegraphics[keepaspectratio=true,scale=0.378]{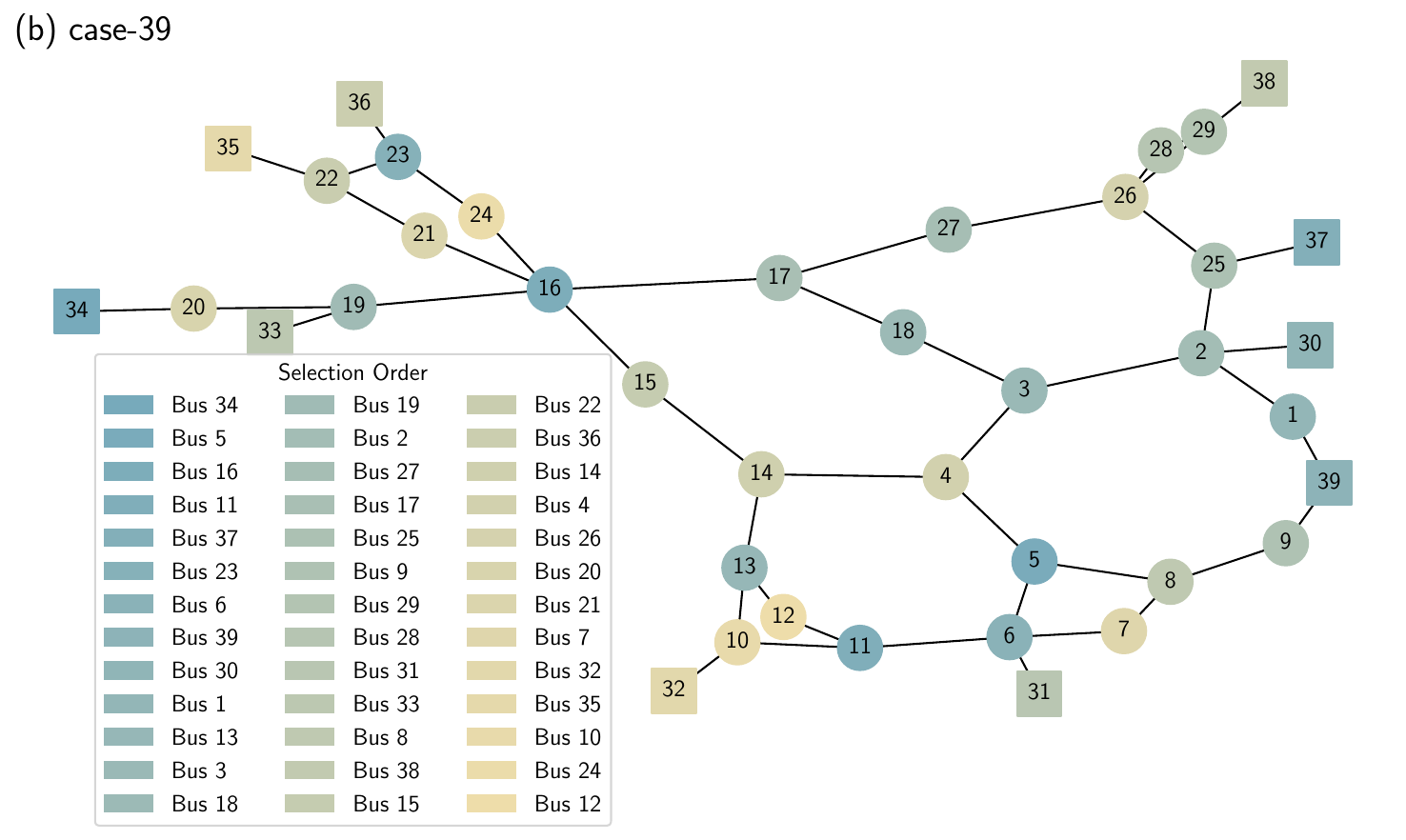}}{}{} \hspace{-4cm}
	\subfloat{\includegraphics[keepaspectratio=true,scale=0.36]{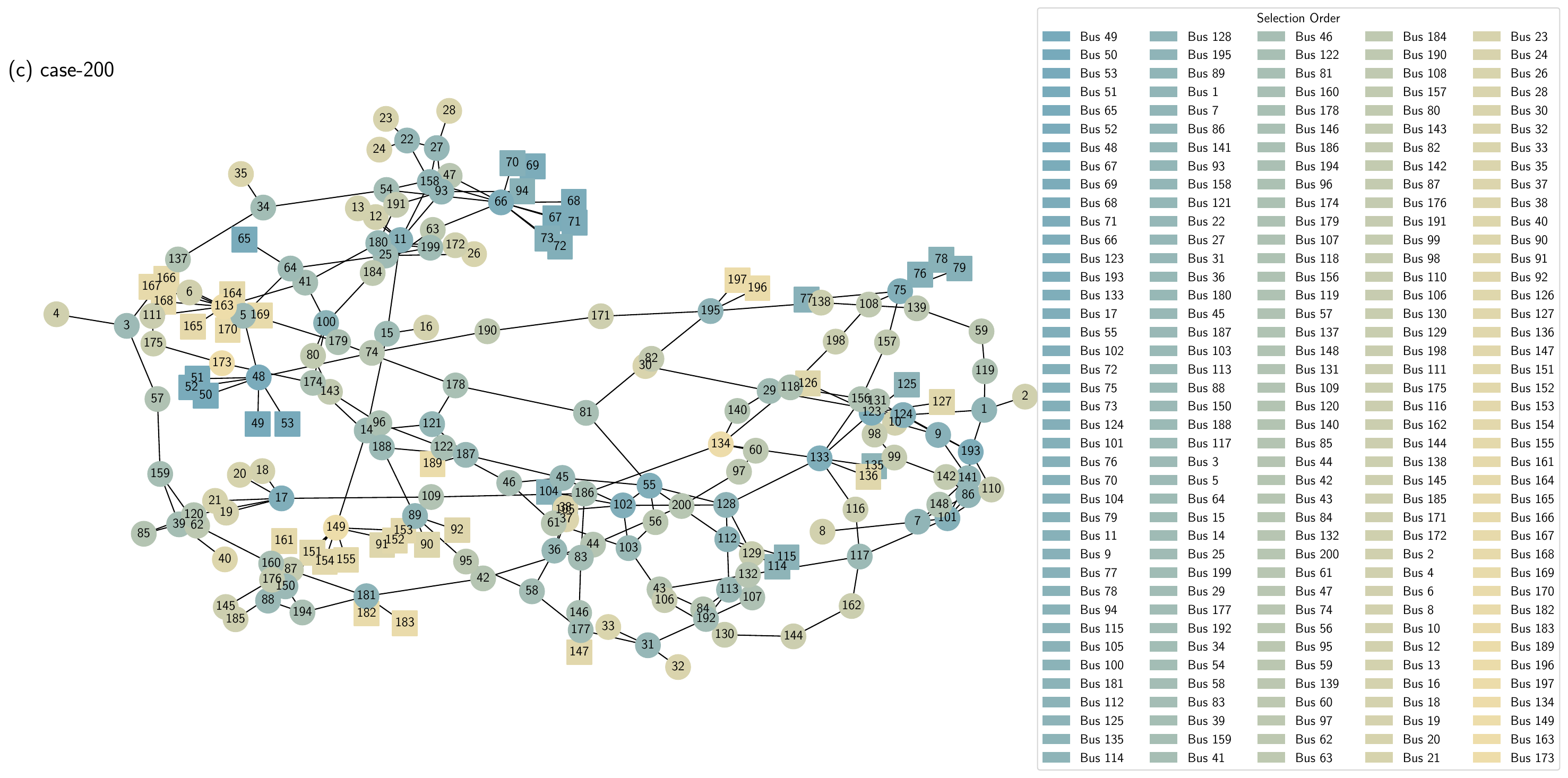}}\vspace{-0.2cm}
	\caption{Nodes from most to least stable obtained by solving $\mr{\textbf{P1}}$ for networks (a) $\mr{case}$-$\mr{9}$, (b) $\mr{case}$-$\mr{39}$ and (c) $\mr{case}$-$\mr{200}$. The square nodes represent generator buses {while the circles represent loads/renewables buses.}}\label{fig:AllocCase}
	\vspace{-0.7cm}
\end{figure*}
\setlength{\textfloatsep}{0pt}

To quantify the overall stability impact on a power network after a renewable injection is allocated at a bus (whether a generator or load bus; see Proposition~\ref{prs:stability-ident}), we solve the optimal problem $\mr{\textbf{P1}}$. The  horizon window for LEs computations is chosen as $\mr{N} = t/h= 300.$ First, we quantify the {parameterized} deformation matrix $\tilde{\m{\Xi}}_i(\m\gamma)$ (see, Proposition~\ref{property:LyapProp}) for each $i \in \mc{N}$. Here, a perturbed renewable load injection along with its uncertainty is applied to a node within the network as described above. To obtain accurate computations of $\tilde{\m{\Xi}}_i(\m\gamma)$ (see, Section~\ref{subsec:LyapTheory}), we utilize Algorithm~\ref{algorithm:DQR}. Then, we utilize Theorem~\ref{theo:logdet} to compute the objective function of $\mr{\textbf{P1}}$, such that the $\mr{log\,det} \left(	\tilde{\m{\Xi}}(\mc{S},\m\gamma)\right)$ is equivalent to computing the sum of LEs of the system after a renewable injection is applied at a node $i \in \mc{N}$. Based on whether a node is a generator or load bus, we obtain the LEs be referring to Proposition~\ref{prs:stability-ident}. The greedy algorithm is then adopted to solve $\mr{\textbf{P1}}$ for $|\mc{S}|=s$. Here we want to quantify the impact of stability when a renewable load is applied at each of the network buses, as such we set $s=N$. The solution to the optimal allocation problem, denoted by $\mr{\mathbf{P1}}$, is the set $\mathcal{S}^*$, where $\mathcal{S}^*$ represents an ordered set of nodes. The nodes are ranked based on their ability to minimize uncertainty propagation in the system {when uncertain renewable loads are allocated at the nodes.} This detailed framework is summarized in Algorithm~\ref{alg:algorithm_allocation}.
\subsection{Stability index and node ranking for potential RER allocation}\label{subsec:LyapExp_Num}
As mentioned in the aforementioned sections, the validity of the NL-DAE model and the discretization is studied in~\cite{Kazma2023b}. With that in mind, we now want to analyze the stability of the nodes through quantifying the LEs of the network that already has renewable load injections allocated at random buses. The LEs for $\mr{case}$-$\mr{9}$ and $\mr{case}$-$\mr{200}$ computed using~\eqref{eq:specLyapExpQR}, considering all the differential and algebraic states, are depicted in left column of Fig.~\ref{fig:LyaExp-9}. The states of the systems are summarized in~\eqref{diffstate}--\eqref{algstate}. For brevity, the LEs for $\mr{case}$-$\mr{39}$ are omitted. Notice that the LEs corresponding to each state range from positive LEs to negative LEs indicating that the system is inundated with transients from the allocated renewables. The LEs that characterize a bus stability are then computed for the buses of the two systems, according to Proposition~\ref{prs:stability-ident}, by computing the parameterized state-deformation matrix; see, Proposition~\ref{prs:param_Lyapexp}. The corresponding LEs are computed and then ordered to obtain the stability index of each bus in the system denoted as $\mr{S}_i$ for $i\in\mc{N}$. The right column of Fig.~\ref{fig:LyaExp-9} depicts the ordered LEs that are equivalent to the stability index as discussed in Section~\ref{subsec:StableNodeIdent}. {For $\mr{case}$-$\mr{9}$, buses $\{1, 2, 3\}$ are generator buses; their stability index depends on frequency, voltage or rotor angle.} This is a consequence of Property~\ref{property:LyapProp} that is discussed in Remark~\ref{rmk:Lyap_properties}. {Similarly, for $\mr{case}$-$\mr{200}$, the stability quantification of the buses can result from instability related to frequency, voltage or rotor angle as depicted in Fig.~\ref{fig:LyaExp-9}.} These answers question $\mathrm{Q}1$ posed at the beginning of this section.

 \begin{figure}[t]
	\centering
	\vspace{-0.2cm}
	\hspace{-0.3cm}
	\subfloat{\includegraphics[keepaspectratio=true,scale=0.39]{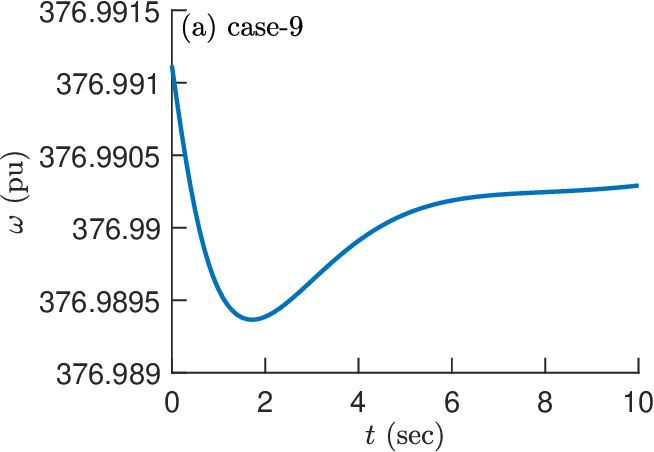}}\vspace*{-0.35cm} \hspace{-0.1cm}
	\subfloat{\includegraphics[keepaspectratio=true,scale=0.39]{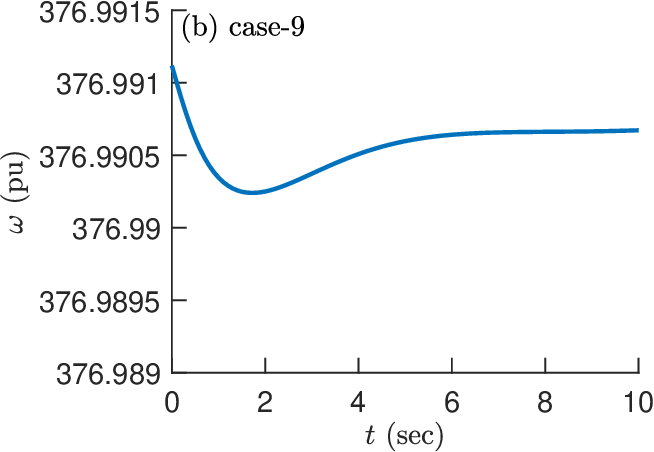}}{}{}
	\hspace{-1cm}
	\subfloat{\includegraphics[keepaspectratio=true,scale=0.39]{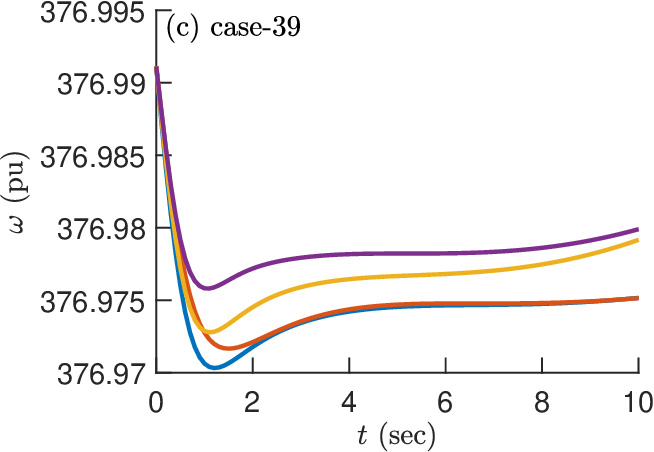}}\vspace*{-0.35cm} \hspace{-0.1cm}
	\subfloat{\includegraphics[keepaspectratio=true,scale=0.39]{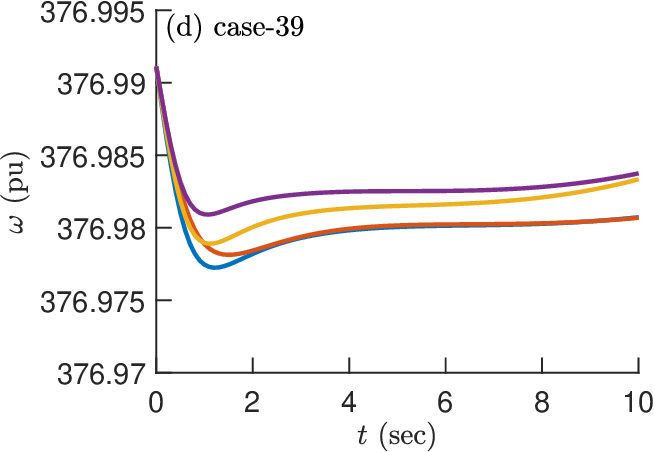}}{}{}
	\subfloat{\includegraphics[keepaspectratio=true,scale=0.39]{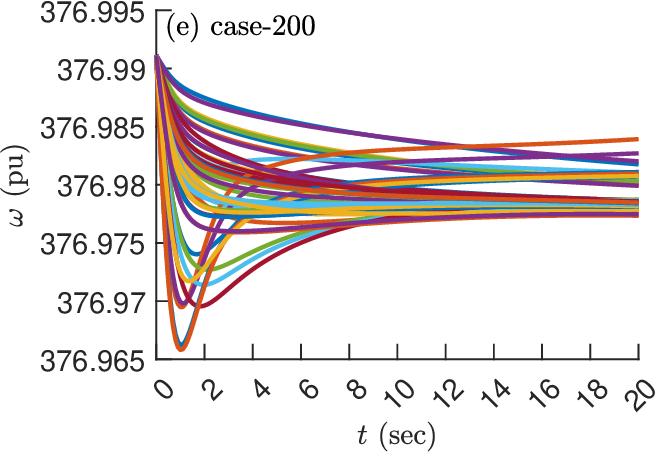}}
	\subfloat{\includegraphics[keepaspectratio=true,scale=0.39]{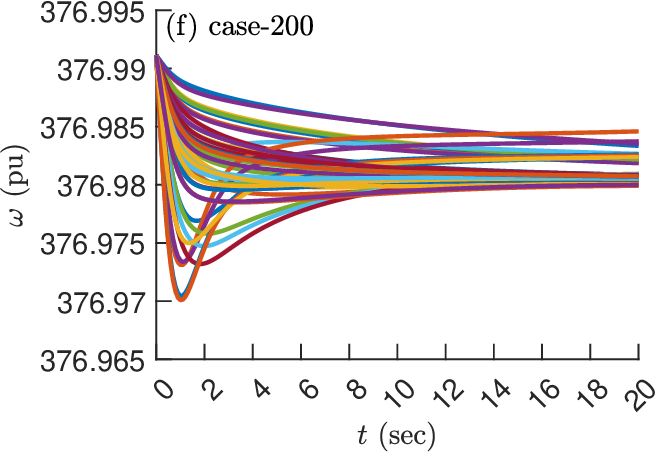}}{}{}
	\vspace{-0.2cm}
	\caption{Clearing time for frequency transients induced by allocating an uncertain renewable load injection at a random bus (a, c, e) and at the most stable bus (b, d, f), computed by solving $\mr{\textbf{P1}}$ for each case system.}\label{fig:response}
\end{figure}

Based on the validity of stability quantification method resulting from the computation of the system's spectrum of LEs, we now solve $\mr{\textbf{P1}}$. The optimal uncertain renewable load allocation problem is solved according to Algorithm~\ref{alg:algorithm_allocation} as discussed in Section~\ref{sec:OPP_Frame}. The ranked nodes (buses) denoted by $\mc{S}^*$ are depicted in Fig.~\ref{fig:AllocCase}. The buses are ranked from most stable to least stable, thereby indicating which nodes allow for allocating a renewable load while having the least impact on the overall stability of the power network. {For $\mr{case}$-$\mr{9}$, the optimal bus location is a generator, Bus No. $2$. This is followed by a load, Bus No. $4$, and a generator, Bus No. $1$.} Whereas by referring to Fig.~\ref{fig:LyaExp-9}, notice that the top three stable nodes, in order, are $\{1,2,4\}$. This is due to the following: when computing the stability index, we quantify the stability of a node based on the LEs resulting from frequency, voltage, or rotor angle instability. However, for the allocation problem when solving  $\mr{log\,det} \left(	\tilde{\m{\Xi}}(\mc{S},\m\gamma)\right)$, we sum all the LEs of all the buses within the network. That is, we are quantifying the impact of a renewable load injection on the stability of all the nodes by computing all the LEs of the system after the uncertain load is allocated to the power grid. As such, the stability index of the buses can be different than the stable nodes identified by solving problem $\mr{\textbf{P1}}$. Note that the plotted networks are based on branch data extracted from MATPOWER case files for each of the studied networks.

 For $\mr{case}$-$\mr{39}$, the most stable generator nodes are buses $\{34,37,39\}$, and the most stable load nodes are buses $\{5,16,11\}$, as depicted in Fig.~\ref{fig:AllocCase}. Notice that for the load buses, the three buses are included in 3 different community structures within the power network. {Here, we refer to a sub-network as a community structure, where several nodes are connected and form a loop connecting to a neighboring community.} As mentioned in the introduction, in this paper we approach quantifying the impact of a node on overall system stability from a dynamical perspective. In the literature, certain edge and node centrality measures can be used to assess the impact of loads at specific nodes, however, such approaches are graph-theoretic. A possible extension to the work related to this study, is to integrate such centrality measures with the developed stability measure herein. The results for $\mr{case}$-$\mr{200}$ are also depicted in Fig.~\ref{fig:AllocCase}. The full ranking of the power network buses is shown from most stable to least stable. The identification shows that certain sub-networks are less stable than others, and thus the allocation of RERs is preferable within those communities.
  
 
 Now, in order to validate the optimality of the ranked nodes $\mc{S}^*$ for each of the test systems, we simulate the dynamics by allocating an uncertain RER at the most stable node and compare to a system dynamics scenario that has a random RER allocation. The frequency response of the generators for $\mr{case}$-$\mr{9}$, $\mr{case}$-$\mr{39}$ and $\mr{case}$-$\mr{200}$ are shown in Fig.~\ref{fig:response}. Notice that the transients clear faster for the case when the same renewable load injections are applied to the optimal stable node (right column) for both test systems. The clearing time for the transients for $\mr{case}$-$\mr{9}$ is $0.5$ seconds faster and has a damped transient as compared to the random node. {For $\mr{case}$-$\mr{39}$ and $\mr{case}$-$\mr{200}$, the transient responses are also damped and are around $1$ second shorter.} The resulting faster clearing time indicates that uncertainty from renewable load injections can propagate slower depending on where the RER is allocated. This also demonstrates the optimality of the stability result obtained from solving $\mr{\textbf{P1}}$, thereby answering questions $\mathrm{Q}2$  and $\mathrm{Q}3$.
 
 Finally, we numerically verify Theorem~\ref{theo:logdet}. In doing so, we evaluate the $\mr{log\,det}$ of matrix $\tilde{\m{\Xi}}_i(\m\gamma)$ and compare it to the sum of the spectrum of LEs computed using~\eqref{eq:specLyapExpQR}. For $\mr{case}$-$\mr{9}$ and $\mr{case}$-$\mr{200}$ we obtain an equivalence relation as presented in Theorem~\ref{theo:logdet}. For $\mr{case}$-$\mr{9}$ the sum of all LEs is equal to $0.2349$ and for $\mr{case}$-$\mr{200}$ the sum of LEs is $2.6334$. The results thus provide a clear relation between the stability measure used in the optimal renewable load allocation problem $\mr{\textbf{P1}}$ and the LEs exponents that quantify system stability. The equivalence, therefore, answers question $\mathrm{Q}4$.
 \subsection{PV and load-integrated power system}\label{subsec:RER_allocation}
In this section, we consider an NL-DAE power system that details the dynamic models of synchronous generators, grid-forming PV plants, motor loads, and dynamic loads, as detailed in Appendix~\ref{apndx:power_model2}. All parameters for the motor load and synchronous generators, including the excitation system, can be found in~\cite{Sauer2017}, while a comprehensive description of the PV model along with its parameters is provided in~\cite{Roy2023}. The modified WSCC power system $\mr{case}$-$\mr{9PV}$ is represented in Fig.~\ref{fig:case9PV}. The synchronous speed is set to $\omega_{0} = 120\pi \mr{rad/sec}$, while the power base is chosen as $100 \mr{MVA}$. Similar to the first model, the steady-state values before any load perturbation are solved using power flow studies with the $\mr{runpf}$ function.
 
 For this power network, we induce a renewable load perturbation by changing the irradiance on a PV plant. The two PV plants $\mr{S}1$ and $\mr{S}2$ operate under a solar irradiance $I_r^0 = 1000 \mr{~W} / \mr{m}^2$. The perturbations induce a change in the irradiance value for a PV plant as $I_r^e=\left(1-\Delta_I\right)\left(I_r^0\right)+q_I$, where $ I_r^e$ denotes the irradiance after the disturbance, $\Delta_I$ is the severity of the disturbance, {$q_I$ is a Gaussian noise with zero mean and variance of $0.01 \Delta_I$.}
 
\begin{figure}[t]
	\centering
	\subfloat{\includegraphics[keepaspectratio=true,scale=0.38]{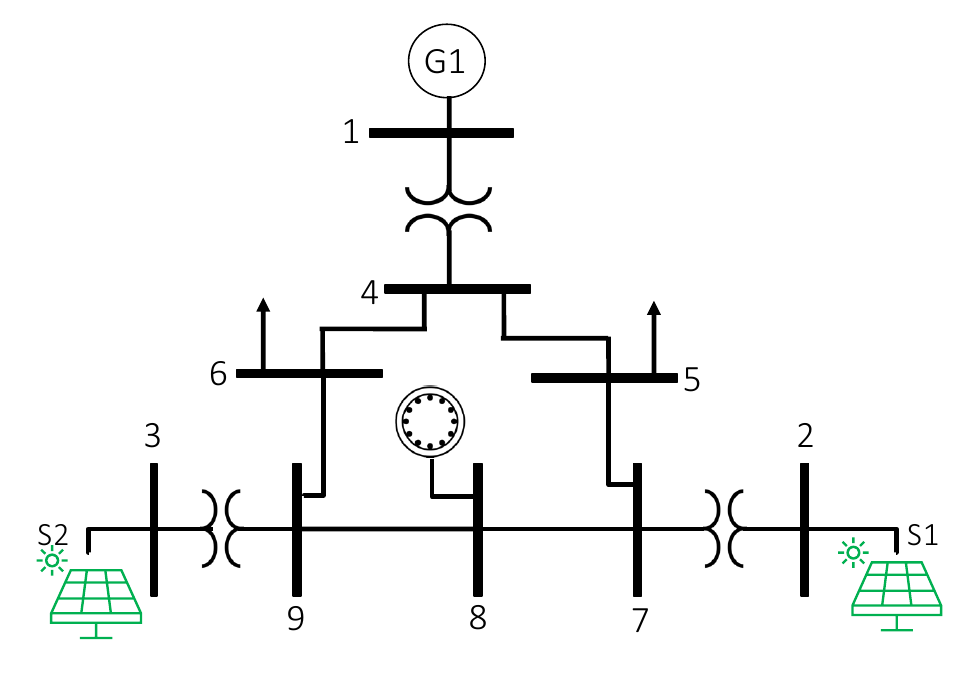}}
	\vspace{-0.4cm}
	\caption{One-line diagram of the modified WSCC power system~\cite{Roy2023} ($\mr{case}$-$\mr{9PV}$). It includes a motor load at Bus 8, a synchronous generator at Bus 1, and two PV plants $\mr{S} 1$ and $\mr{S} 2$ at Buses 2 and 3.}\label{fig:case9PV}
\end{figure}
\setlength{\textfloatsep}{0pt}

Based on the validity of uncertain renewable load allocation problem when applied to the $4^{\text{th}}$-order system, we now solve the allocation problem for the comprehensive power grid detailed in Appendix~\ref{apndx:power_model2}. We utilize the same notions of stability for PV plants, by considering voltage stability. That is, we analyze the voltages $v$ $(pu)$ across the AC capacitor along the dq-axis. While we could consider stability with respect to the inverter angles ${\delta}$ $(pu)$ of the PV plant, for consistency with Proposition~\ref{prs:stability-ident}, we consider voltage stability. The parameterized tensor matrix presented in Proposition~\ref{prs:param_Lyapexp} allows for the {depiction of perturbations along the trajectory of any selected state within the dynamical system.} Meaning that for more complex systems with different types of buses, one can choose the states to be considered when computing overall system stability.
 \begin{figure}[t]
 	\centering
 	\includegraphics[keepaspectratio=true,scale=0.5]{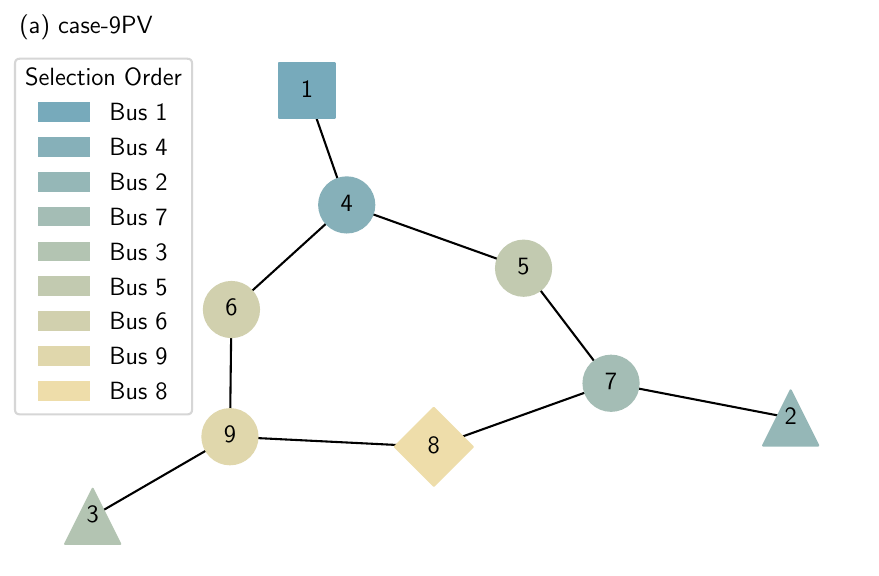}
 	\caption{The ordered nodes from stable to least stable for the $\mr{case}$-$\mr{9PV}$ network. The triangular nodes represent PV plants, while the diamond represents the motor load.}\label{fig:AllocCase9PV}
 \end{figure}
 \setlength{\textfloatsep}{0pt}
 That being said, the ranked nodes for $\mr{case}$-$\mr{9PV}$ are depicted in Fig.~\ref{fig:AllocCase9PV}. {The ranking shows that the most stable generator node is Bus No. $1$, and the most stable load is Bus No. $4$, which is connected to Bus No.  $1$.} The least stable node is Bus No.$8$, which is a motor load bus. This is expected since motors are sensitive to variations in voltages due to lack of voltage control. PV plants $\mr{S}1$ and $\mr{S}2$ connected to buses No. $2$ and $3$, exhibit good overall stability as compared to load buses within the network. This is due to the inverter control that regulates voltage and power, which allow for a reduced uncertainty propagation within the network. Similar to the previous power grid, we assess the validity of the ranked system buses by analyzing the impact of applying an uncertain RER perturbation at both a random bus and the most stable bus; see Fig.~\ref{fig:response2}.
  \begin{figure}[b]
 	\centering
	\vspace{-0.2cm}
 	\hspace{-0.3cm}
 	\subfloat{\includegraphics[keepaspectratio=true,scale=0.39]{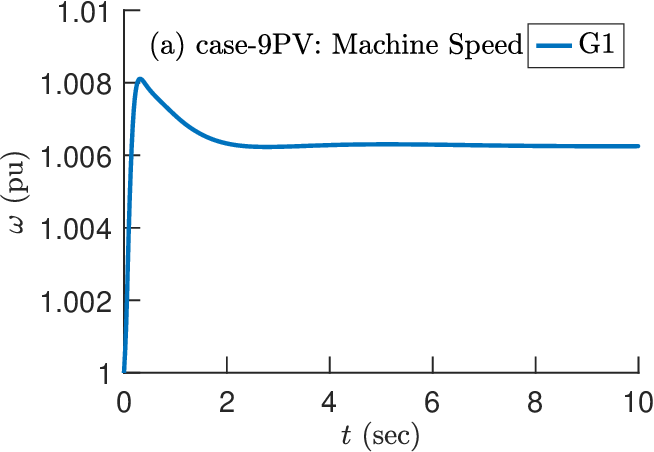}}\vspace*{-0.2cm} \hspace{-0.1cm}
 	\subfloat{\includegraphics[keepaspectratio=true,scale=0.39]{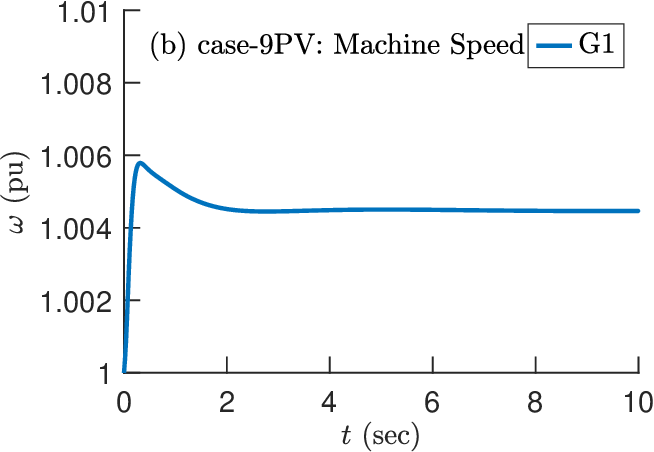}}{}{}
 	\subfloat{\includegraphics[keepaspectratio=true,scale=0.39]{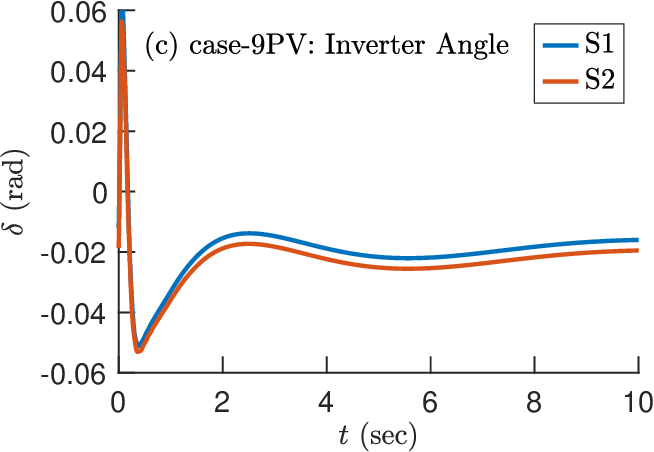}}\vspace*{-0.1cm} \hspace{-0.1cm}
	\subfloat{\includegraphics[keepaspectratio=true,scale=0.39]{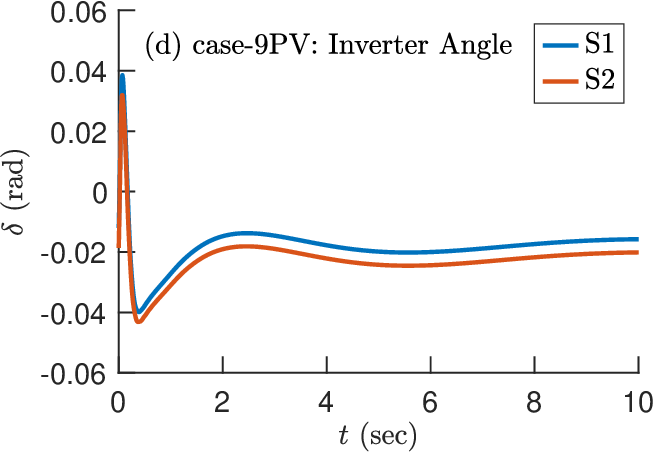}}{}{}
 	\caption{Frequency transients of the generator and inverter angles of the PV plants induced by allocating an uncertain renewable load injection at a random bus (a, c) and the most stable bus (b, d).}\label{fig:response2}
 \end{figure}
 The frequency response of Bus No. 1 and the inverter angle transients are depicted in Fig.~\ref{fig:response2}. Notice that the transients exhibit lower amplitude and shorter clearing time when applying a load perturbation to the stable node as compared to a random network node. This verifies the optimality of $\mathcal{S}^{*}$ obtained by solving $\mathbf{P1}$ for the $\mr{case}$-$\mr{9PV}$ power network.
 
We now attempt to answer $(\mathrm{Q}5)$ by considering the optimal nodes obtained from solving $\mathbf{P1}$ for the $\mr{case}$-$\mr{9PV}$ power network with PV plants. The modified network has a motor load at Bus No.$8$, a synchronous generator at Bus No. $1$, and two PV plants $\mr{S} 1$ and $\mr{S} 2$ at buses No. $2$ and $3$. Based on computing the overall stability implications resulting from allocating a perturbed RER load, the most stable location is found to be Bus No. 1 (see Fig.~\ref{fig:AllocCase9PV}), while the least stable is Bus No. 8. This result informs the system operator that allocating an RER at a motor load bus can result in the largest uncertainty propagation within the network from an overall stability perspective. In contrast, Bus No. 1 results in attenuated transients when allocated with the same RER under the same perturbed loads. The system operator can therefore observe which types of buses, along with their locations in a given network, can handle the uncertain RER load injections from allocating such resources within the grid. Observing the result and based on quantifying frequency, rotor angle, and voltage stability, we observe that the generator bus best attenuates the uncertainty from RER loads, while motor load buses lack this ability. The PV plants, as seen in Fig.~\ref{fig:AllocCase9PV}, can also handle the addition of uncertain loads due to their ability to regulate voltage and current using a proportional-integral (PI) type controller. We emphasize here that system operators can choose to parameterize the computation of LEs, presented in Proposition~\ref{prs:stability-ident}, by considering a single stability aspect, such as voltage stability. Thus, instead of focusing on the effect of allocating an RER on overall system stability, the operator can choose to allocate an RER based on its implications for voltage stability. This can be explored in future work, which may include developing scenarios for RER load dispatch based on different system stability aspects. Furthermore, the parameterization of state-deformation matrix (Proposition~\ref{prs:param_Lyapexp}) along any system state allows to consider aspects regarding the stability of RER resources, such as PV plants, by considering resonance stability and inverter angle stability. On such note, we end this section.
\section{Conclusions, limitations and future work}\label{sec:summary}
This paper presents a framework for quantifying overall power grid stability after a perturbed renewable load is allocated to a bus within the grid. The method is based on computing the spectrum of LEs of a system; it is based on stability criteria that consider frequency, voltage and rotor angle stability. {The proposed method allows for quantifying renewable uncertainty propagation from a dynamical systems perspective across the entire network.} As such, it informs the operator/utility where to practically allocate renewables while maintaining the overall stability of the power grid. This paper is not devoid of limitations. From a theoretical perspective, we focus on quantifying network stability based on frequency, rotor angle, and voltage. However, for RER-integrated networks, other stability measures can be considered. From a practical perspective, we consider two networks, the second is applied to a 9 bus modified network with two PV plants. Future work will focus on incorporating dynamic models of RERs, such as wind farms, and expanding the analysis to larger power systems. That being said, this is the first attempt at quantifying uncertainty propagation and its overall stability implications for model-based NL-DAE power systems. This enables modeling the uncertainty from renewable loads and therefore capturing the full uncertainty propagation within the nodes of the power system.







\bibliographystyle{ieeetr}%
\bibliography{library}

%
%
%
\appendices
\section{$4^{\text{th}}$-order power system model}\label{apndx:power_model} 
This system depicts the standard two axes $4^{\text{th}}$-order transient model of a synchronous generator~\cite[Ch. 7]{Sauer2017}. The considered model excludes exciter dynamics and turbine governor, meaning that each of the machines has four states and two control inputs. We note that the controller (i.e., governor response) adjusts the active power produced by a generator to stabilize the system after a disturbance from renewables. For a synchronous generator $i \in \mc{G}$; its $4^{\text{th}}$-order differential dynamics can be written as
\begin{subequations}\label{eq:differential_dynamics}
	\begin{align}
		&\dot{\delta_{i}} =  \omega_{i} - \omega_{0},\label{eq:differential_dynamics_1}\\
		&M_i\dot{\omega}_{i} = T_{\mr{N}i} - P_{\mr{G}i} - D_{i}(\omega_{i}-\omega_{0}),\label{eq:differential_dynamics_2}\\
		&T^{'}_{\mr{d0}i} \dot{E^{'}_{i}} = -\frac{x_{\mr{d}i}}{x^{'}_{\mr{d}i}}E^{'}_{i} + \frac{x_{\mr{d}i}- {x_{\mr{d}i}^{'}}}{x^{'}_{\mr{d}i}}v_{i}\cos(\delta_{i}- \theta_{i}) + E_{\mr{fd}i},\label{eq:differential_dynamics_3}\\
		&T_{\mr{CH}i}\dot{T}_{\mr{N}i} = T_{\mr{N}i} - \frac{1}{R_{\mr{D}i}}(\omega_{i}-\omega_{0}) + T_{\mr{r}i},\label{eq:differential_dynamics_4}
	\end{align}
\end{subequations}
where the time varying components in~\eqref{eq:differential_dynamics} are: $\delta_{i}$ the rotor angle $\mr{(rad)}$, $\omega_{i}$ generator rotor speed $\mr{(rad/sec)}$, ${E}^{'}_{i}$ generator transient voltage $\mr{(pu)}$, ${T}_{\mr{N}i}$ generator mechanical torque $(\mr{pu})$. Generator inputs are: $E_{\mr{fd}i}$ generator internal field voltage $(\mr{pu})$, $T_{\mr{r}i} $ governor reference signal $(\mr{pu})$. 
Constants in~\eqref{eq:differential_dynamics} are: $M_{i}$ is the rotor inertia constant $(\mr{pu} \times \mr{sec}^2)$, $D_{i}$ is the damping coefficient $(\mr{pu} \times \mr{sec}^2)$, $x_{\mr{d}i}$ and $x_{\mr{q}i}$ are the direct-axis synchronous reactance $(\mr{pu})$, $x^{'}_{\mr{d}i}$ is the direct-axis transient reactance $(\mr{pu})$, $T^{'}_{\mr{d0}i}$ is the direct-axis open-circuit time constant $\mr{(sec)}$, $T_{\mr{CH}i}$ is the chest valve time constant $\mr{(sec)}$, $R_{\mr{D}i}$ is the speed governor regulation constant $\mr{(Hz/pu)}$, and $\omega_0$ is the synchronous speed $(120\pi \;  \mr{rad/sec})$. 

The algebraic constraints of the power system represent the relation between the internal states of a synchronous generator, and it's generated power $P_{\mr{G}i}$ and $Q_{\mr{G}i}$, i.e., real and reactive power. The algebraic constraints of the NL-DAE system can be written as~\eqref{eq:algebriac_constraints} with $i \in \mc{G}$
\begin{subequations}\label{eq:algebriac_constraints}
	\begin{align}
				\hspace{-0.3cm} P_{\mr{G}i} =& \tfrac{1}{x^{'}_{\mr{d}i}}E^{'}_{i}v_{i}\sin(\delta_{i}-\theta_{i})  - \tfrac{x_{\mr{q}i}-x^{'}_{\mr{d}i}}{2x^{'}_{\mr{d}i}x_{\mr{q}i}}v^{2}_{i}\sin(2(\delta_{i}-\theta_{i}))\label{eq:algebriac_constraints_1}\\
				\begin{split}
							\hspace{-0.3cm} Q_{\mr{G}i} =& \tfrac{1}{x^{'}_{\mr{d}i}}E^{'}_{i}v_{i}\cos(\delta_{i}-\theta_{i}) - \tfrac{x_{\mr{q}i}-x^{'}_{\mr{d}i}}{2x^{'}_{\mr{d}i}x_{\mr{q}i}}v^{2}_{i}\\
							&-\tfrac{x_{\mr{q}i}-x^{'}_{\mr{d}i}}{2x^{'}_{\mr{d}i}x_{\mr{q}i}}v^{2}_{i}\cos(2(\delta_{i}-\theta_{i})).
						\end{split}\label{eq:algebriac_constraints_2}
			\end{align}	
\end{subequations}
where $\theta_{ij} = \theta_{i} - \theta_{j}$ is the bus angle, $v_i$ is the bus voltage $\mr{(pu)}$. 

The power balance between the set of generator and load buses with $i \in \mc{G} \cup \mc{L}$ can be written as~\eqref{eq:power_balance}. The power balance in~\eqref{eq:power_balance} resembles the power transfer between RER, generators and loads as follows 
\begin{subequations}\label{eq:power_balance}
	\begin{align}
		P_{\mr{G}i} + P_{\mr{L}i} +  P_{\mr{R}i} &= 
		\sum_{j=1}^{N} v_{i}v_{j}(G_{ij}\cos\theta_{ij}+B_{ij}\sin\theta_{ij}),\\
		Q_{\mr{G}i} +Q_{\mr{L}i} +Q_{\mr{R}i} &= 
		\sum_{j=1}^{N} v_{i}v_{j}(G_{ij}\cos\theta_{ij}-B_{ij}\sin\theta_{ij}),\vspace*{-0.2cm}
	\end{align} 
\end{subequations}
where matrices $G_{ij}$ and $B_{ij}$ denote, respectively, the conductance and susceptance between bus $i$ and $j$. It is noteworthy to mentioned that the load injections that model RER are included in the power balance equations~\eqref{eq:power_balance} as real and reactive power $P_{\mr{R}i}$ and $Q_{\mr{R}i}$. Accordingly, the differential state, algebraic state and input vectors are summarized as follows

\vspace{-0.2cm}
\begin{subequations}
	\small
	\begin{align}
			\m{x}_{d} &=  \left\{ \{\delta_{i}\}_{i=0}^{G}, \; \{\omega_{i}\}_{i=0}^{G},\; 
			 \{E^{'}_{i}\}_{i=0}^{G}, \;  \{{T}_{\mr{N}i} \}_{i=0}^{G} \right\}^{\top} \in \Rn{n_d}, \label{diffstate}\\
			\m{x}_{a} &= \left\{\{{P}_{\mr{G}i} \}_{i=0}^{G},\; \{{Q}_{\mr{G}i} \}_{i=0}^{G},\; \{{v}_{i} \}_{i=0}^{N},\; \{{\theta}_{i} \}_{i=0}^{N}\right\}^{\top}\in \Rn{n_a},\label{algstate}\\
			\m{u} &= \left[ \{{E}_{\mr{fd}i} \}_{i=0}^{G}, \; \{{T}_{\mr{r}i} \}_{i=0}^{G} \right]^{\top} \in \Rn{n_u}.\label{inputstate} 
		\end{align}
\end{subequations}

\section{PV and load-integrated power system model}\label{apndx:power_model2} 
The comprehensive $9^{\text{th}}$-order transient model of synchronous generator $i \in \mc{G}$, that includes the governor, hydro/steam turbine and IEEE type DC1 excitation system dynamics can be given as follows.

\hspace{-0.21cm}\textbullet\hspace{+0.21cm} The swing equations for $i \in \mc{G}$ can be written as follows
\begin{subequations}\label{eq:powermodel2-1}
\begin{align}
	& \dot{\delta_{i}}=\omega_i-\omega_0, \\
	& \dot{\omega}_{i}=\frac{1}{2 H_i}\left(T_{\mathrm{M}_i}-T_{\mathrm{e}_{\mathrm{i}}}\right), \quad T_{\mathrm{e}_i}=E_{\mathrm{d}_i} i_{\mathrm{d}_i}+E_{\mathrm{q}_i} i_{\mathrm{q}_i}, \\
	& \dot{E}_{\mathrm{q}_i}=-\frac{1}{t_{\mathrm{qo}_i}}\left(E_{\mathrm{q}_i}-\left(x_{\mathrm{q}_i}^{\prime}-x_{\mathrm{q}_i}\right) i_{\mathrm{d}_i}\right), \\
	& \dot{E}_{\mathrm{d}_i}=-\frac{1}{t_{\mathrm{do}_i}}\left(E_{\mathrm{d}_i}+\left(x_{\mathrm{d}_i}^{\prime}-x_{\mathrm{d}_i}\right) i_{\mathrm{q}_i}-E_{\mathrm{fd}_i}\right),
\end{align}
\end{subequations}
where the time-varying components in ~\eqref{eq:powermodel2-1} are: $\delta_i$ is the generator rotor angle $(\mr{pu}), \omega_i$ is the generator rotor speed $(\mr{pu})$, $E_{\mr{q}_i}$ and $E_{\mr{d}_i}$ are the dq-axis voltages $(\mr{pu})$ for transient reactance. Constants in~\eqref{eq:powermodel2-1} are: $\omega_0$ which denotes the synchronous speed $(120\pi \;  \mr{rad/sec})$, $x_{\mathrm{q}_i}^{\prime}, x_{\mathrm{d}_i}^{\prime}, x_{\mathrm{q}_i}, x_{\mathrm{d}_i}$ are the synchronous transient reactance $(\mr{pu})$ along dq-axis, $i_{\mr{d}_i}, i_{\mr{q}_i}$ are the dq-axis generator currents, $t_{\mr{qo}_i}, t_{\mr{do}_i}$ are the open circuit time constants $(\mr{sec})$ along dq-axis, $H_i$ is the generator inertia constant $(\mr{pu} \times \mr{sec})$, and $T_{\mr{e}_i}$ is the electrical air gap torque $(\mr{pu})$.

\hspace{-0.21cm}\textbullet \hspace{+0.21cm}The turbine and governor dynamics can be written as follows
\begin{subequations}\label{eq:powermodel2-2}
\begin{align}
	&\hspace{-0.27cm} \dot{T}_{\mathrm{M}_i}= \begin{cases}-\frac{1}{t_{\mr{ch}i }}\left(T_{\mathrm{M}_i}-P_{v_i}\right), & \text { for thermal, } \\
		-\frac{2}{t_{w i}}\left(T_{\mathrm{M}_i}-P_{v_i}+t_{\mathrm{ch} i} \dot{P}_{v_i}\right), & \text { for hydro, }\end{cases} \\
	& \hspace{-0.27cm}\dot{P}_{v_i}=-\frac{1}{t_{v i}}\left(P_{v_i}-P_{v_i}^*+\frac{\omega_i-1}{R_{d i}}\right),
\end{align}
\end{subequations}
where the time-varying components in ~\eqref{eq:powermodel2-2} are: $T_{\mr{M}_i}$ which represents the mechanical torque input $(\mr{pu})$ and $P_{{v}_{i}}$ which denotes the steam/hydro valve position $(\mr{pu})$.
	The constants in~\eqref{eq:powermodel2-2} are:  $R_{\mr{d}i}$ which denotes the governor droop constant $(\mr{Hz} / \mr{pu})$, and $ P_{v_i}^*$ which denotes the valve position set point $(\mr{pu})$ from the grid operator.
	
\hspace{-0.21cm}\textbullet \hspace{+0.21cm}The excitation system dynamics can be written as follows
\begin{subequations}\label{eq:powermodel2-3}
	\begin{align}
	&\hspace{-0.25cm}\dot{E}_{\mr{fd}_i} =\tfrac{-1}{t_{\mathrm{fd}i}}\left(k_{e i}+S_{e i} E_{\mathrm{fd}_i}-v_{a i}\right), \quad S_{e i}=a_i e^{b_i E_{\mathrm{fdi}}}, \\
&\hspace{-0.25cm}	\dot{r}_{f_i} =-\tfrac{1}{t_{f i}}\left(r_{f_i}-\frac{k_{f i}}{t_{f i}} E_{\mathrm{fd}_i}\right), \\
&\hspace{-0.25cm}	\dot{v}_{a i} =-\tfrac{1}{t_{a i}}\left(v_{a i}-k_{a i} v_{e i}\right),
\end{align}
\end{subequations}
where the time-varying components in ~\eqref{eq:powermodel2-3} are: $E_{\mathrm{fd}}$ which denotes the generator field voltage $(\mr{pu})$, $v_{ai}$ denotes the amplifier voltage $(\mr{pu})$, and $ r_{\mr{f}_i}$ is the stabilizer output $(\mr{pu})$. In~\eqref{eq:powermodel2-3}, $t_{\mr{chi}}, t_{{w}i}, t_{{v}i}, t_{\mr{fd}i}, t_{f i}$, and $t_{ai}$ are the time constants $(\mr{sec})$ for steam/hydro valve position, field voltage, stabilizer, and amplifier. Such that, $S_{e i}$ represents the generator field voltage saturation function with constant scalars $a_i, b_i$. Furthermore,$k_{e i}, k_{f i}, k_{a i}$ are the exciter, stabilizer, and amplifier gains. The voltage set point from the grid operator is denoted as $V_i^*$, where $V_i$ is the generator terminal voltage. Accordingly, the differential state, algebraic state and input vectors are summarized as follows

\vspace{-0.3cm}
\begin{subequations}
	\small
	\begin{align}
		\m{x}_{d} &= \left\{\begin{array}{lll}
			\m{x}_G^{\top} & \m{x}_R^{\top} & \m{x}_m^{\top}
		\end{array} \right\}^{\top} \in \mathbb{R}^{n_d}, \label{diffstate2}\\
		\m{x}_{a} &= \left\{\{I_{\mr{Re}_i} \}_{i=0}^{N},\; \{I_{\mr{Im}_i} \}_{i=0}^{N},\; \{V_{\mr{Re}_i} \}_{i=0}^{N},\; \{V_{\mr{Im}_i} \}_{i=0}^{N}\right\}^{\top}\in \Rn{n_a},\label{algstate2}\\
		\m{u} &= \left[ \{V_i^* \}_{i=0}^{G}, \;\{P_v^* \}_{i=0}^{G},\;\{V_s^* \}_{i=0}^{G}, \;\{P_s^* \}_{i=0}^{G}, \right]^{\top}  \in \Rn{n_u}\vspace{-0.2cm}\label{inputstate2}
	\end{align}
	\end{subequations}
where $\m{x}_G$ are the dynamic states of the $9^{\text{th}}$-order system (synchronous generator, excitation system, governor, and turbine dynamics), $\m{x}_R$ denotes the dynamic states of the PV plant, and $\m{x}_m$ denotes the states of motor loads. Vectors  $\left\{\{I_{\mr{Re}_i} \}_{i=0}^{N},\; \{I_{\mr{Im}_i} \}_{i=0}^{N},\; \{V_{\mr{Re}_i} \}_{i=0}^{N},\; \{V_{\mr{Im}_i} \}_{i=0}^{N}\right\}$ are the real and imaginary parts of current and voltages. Furthermore, ${P}_s^*$ and ${V}_s^*$ are the power (pu) and voltage (pu) reference set points for solar PV plants. For brevity, we refer the reader to~\cite{Roy2023} for the detailed dynamic equations describing all the states of the solar power plants and the dynamical model for motor loads.

\section{Newton-Raphson algorithm}\label{apndx:NR-method}
The NR iterative method ensures state convergence for each time-step $k$ by evaluating the Jacobian of the nonlinear dynamics for each $k$ under a NR iteration $i$. The Jacobian is used to compute increment $\Delta \m{x}_{k}^{(i)}$. The computed increment is then used to update the state vector to the next time-step as $\m{x}^{(i+1)}_{k}= \m{x}_{k}^{(i)} + \Delta \m{x}_{k}^{(i)}$ for each NR iteration $i$. Once a convergence criterion is satisfied, $||\Delta \m{x}_{k}^{(i)}||_2 \leq \eps$, time-step $k$ advances to $k+1$ and iterates for the whole time-span $t$. Here $\eps$ is a small and positive convergence criterion. With that in mind, the iteration NR increment $\Delta \m{x}_{k}^{(i)}$ can be written as
\begin{equation}\label{eq:mu-newton_raph}
	\Delta \m{x}^{(i)}_{k} = \left[\mA_{g}(\m{z}^{(i)}_{k})\right]^{-1}\begin{bmatrix}
		\m{\varphi}(\m{z}^{(i)}_{k})
	\end{bmatrix},
\end{equation}
where $\m{\varphi}(\m{z}^{(i)}_{k},\m{z}_{k-1}):=\m{\varphi}(\m{z}^{(i)}_{k})$ is the discrete-time nonlinear model~\eqref{eq:semi_NDAE_rep-ODE} in implicit form such that $\m{\varphi}(\m{z}^{(i)}_{k}) = 0 \;\forall \; k$. The Jacobian $\m{A}_{g}(\m{z}^{(i)}_{k}) := \begin{bmatrix}
	\tfrac{\partial \m{\varphi}(\m{z}^{(i)}_{k})}{\partial \m{x}_k}
\end{bmatrix} \in \Rn{n\times n}$ is defined as 
\begin{equation}\label{eq:mu-Jac_Newton_Raph}
\small	\m{A}_{g}\hspace{-0.05cm}=\hspace{-0.1cm}
	\begin{bmatrix}
		\eye_{n_d}-\tilde{h} 
		\m{F}_{\m{x}_{d}}(\m{z}^{(i)}_{k},\m{z}_{k-1})
		\hspace{-0.1cm}& \hspace{-0.2cm}-\tilde{h} 
		\m{F}_{\m{x}_{a}}(\m{z}^{(i)}_{k},\m{z}_{k-1})\\
		-\tilde{h} 
		\m{G}_{\m{x}_{d}}(\m{z}^{(i)}_{k},\m{z}_{k-1})\hspace{-0.4cm}& \hspace{-0.1cm}
		\eye_{n_a}-
		\tilde{h} \m{G}_{\m{x}_{a}}(\m{z}^{(i)}_{k},\m{z}_{k-1})
	\end{bmatrix},
\end{equation}
where matrices $\m{F}_{\m{x}_{d}}(\m{z}^{(i)}_{k},\m{z}_{k-1}):= \m{F}_{\m{x}_{d}}(\m{z}^{(i)}_{k})+\m{F}_{\m{x}_{d}}(\m{z}_{k-1}) \in \Rn{n_{d}\times n_{d}}$ and $\m{F}_{\m{x}_{a}}(\m{z}^{(i)}_{k},\m{z}_{k-1}):= \m{F}_{\m{x}_{a}}(\m{z}^{(i)}_{k})+\m{F}_{\m{x}_{a}}(\m{z}_{k-1})$ $\in \Rn{n_{d}\times n_{a}}$ are the Jacobians of~\eqref{eq:disc_ssm_NDAE} with respect to $\m{x}_{d}$ and $\m{x}_{a}$. Matrix $\eye_{n_d}$ is an identity matrix of dimension similar to $\m{F}_{\m{x}_{d}}(\cdot)$. Matrices $\m{G}_{\m{x}_{d}}(\m{z}^{(i)}_{k},\m{z}_{k-1}):=\m{G}_{\m{x}_{d}}(\m{z}^{(i)}_{k})+\m{G}_{\m{x}_{d}}(\m{z}_{k-1}) \in \Rn{n_{a}\times n_{d}}$ and $\m{G}_{\m{x}_{a}}(\m{z}^{(i)}_{k},\m{z}_{k-1}):=\m{G}_{\m{x}_{a}}(\m{z}^{(i)}_{k}) +\m{G}_{\m{x}_{a}}(\m{z}_{k-1}) \in \Rn{n_{a}\times n_{a}}$ are the Jacobians with respect to $\m{x}_{d}$ and $\m{x}_{a}$ and $\eye_{n_a}$ is an identity matrix. The NR iterative method is summarized in Algorithm~\ref{algorithm:NR}. 
\section{Submodular set function maximization}\label{apndx:greedy}
{For any set function, denoted by $\mc{L}(\cdot)$, that is submodular and monotone increasing,} the maximization of the set function computed using the greedy algorithm offers a theoretical worst-case bound according to the following theorem.
\begin{theorem}(\hspace{-0.012cm}\cite{Nemhauser1978}) Let $\mc{L}: 2^V \rightarrow \mathbb{R}$ be a submodular and monotone increasing set function, $\mc{L}^*$ be the optimal solution of $\mr{\textbf{P1}}$ and $\mc{L}^*_{\mathcal{S}}$ be the solution computed using the greedy algorithm. Then, the following performance bound holds true 
	\begin{align*}
		\mc{L}^*_{\mathcal{S}} -\mc{L}(\emptyset) \geq \left(1-\frac{1}{e}\right)\left(\mc{L}^*-\mc{L}(\emptyset)\right), \quad \text{with}\;\;\; \mc{L}(\emptyset) =0,
	\end{align*}
	where $e\approx 2.71828$. 
\end{theorem}
We note that the above bound is theoretical, and generally a greedy approach performs better in practice. In practice when considering a submodular set maximization problem,  it has been shown the performance accuracy achieves a $99\%$ guarantee; see~\cite{Summers2016} and the many references that cite this work. 

\begin{algorithm}[t]
	\caption{Newton-Raphson Algorithm}\label{algorithm:NR} 
	\DontPrintSemicolon
	\textbf{Input:} $\m{x}_0$, $\epsilon$, $\max$ iteration\;
	\textbf{Output:} $\m{x}^* = \m{x}_{k+1} \; \forall \; k \; \in \; \{1,2,\ldots,\mr{N}-1\}$\;
	\textbf{Initialize:} $i \leftarrow 0$, $\m{x}^{(i)} \leftarrow \m{x}_0$\;
	\ForAll{$k \; \in \; \{1,2,\ldots,\mr{N}-1\}$}{
		\While{$||\Delta \m{x}_{k}^{(i)}||_2 $ $> \epsilon$ \textbf{and} $i <\max$ iteration}{
			\textbf{compute:} Jacobian $\m{A}_g(\m{z}^{(i)})$ as in~\eqref{eq:mu-Jac_Newton_Raph}\;
			\textbf{solve for:} $\Delta \m{x}^{(i)}$ as in~\eqref{eq:mu-newton_raph}\;
			\textbf{update:} $\m{x}^{(i+1)} \leftarrow \m{x}^{(i)} + \Delta \m{x}^{(i)}$\;
			\textbf{update:} \textit{error} $\leftarrow \|\Delta \m{x}^{(i)}\|_2$\;
			\If{error $\leq \epsilon$}{
				\Return $\m{x}^{(i+1)} = x^{*}$\;
			}
			\Else{
				$i \leftarrow i + 1$\;\;
				\vspace{-0.4cm}
			}
			\vspace{-0.05cm}
		}
		\vspace{-0.05cm}
	}
\end{algorithm}
\vspace{-0.4cm}
\setlength{\textfloatsep}{0pt}

\end{document}